\documentclass[a4,12pt]{article}
\usepackage{amsmath,amsthm}
\usepackage[english]{babel}
\usepackage{amssymb,amstext,amsfonts}
\usepackage{dsfont}
\newtheorem{proposition}{Proposition}[section]
\newtheorem{theo}{Theorem}[section]
\newtheorem{lem}[theo]{Lemma}

\theoremstyle{definition}
\newtheorem{rem}[theo]{Remark}
\newtheorem{defi}[theo]{Definition}

\makeatletter
\@addtoreset{equation}{section}
\makeatother
\usepackage{color}

\begin{document}

\title{On ergodic states, spontaneous symmetry breaking and the Bogoliubov quasi-averages}

\author{Walter F. Wreszinski\\
        Instituto de Fisica USP\\
        Rua do Mat\~{a}o, s.n., Travessa R 187\\
        05508-090 S\~{a}o Paulo, Brazil\\
        \texttt{wreszins@gmail.com}\\
        and\\
        Valentin A. Zagrebnov\\
        Aix-Marseille Universit\'{e}, CNRS, Centrale Marseille, I2M \\
        Institut de Math\'{e}matiques de Marseille - UMR 7373\\
        CMI - Technop\^{o}le Ch\^{a}teau-Gombert\\
        13453 Marseille, France \\
        \texttt{valentin.zagrebnov@univ-amu.fr}\\
        }

\maketitle

\begin{abstract}
{It is shown that Bogoliubov quasi-averages select the pure or ergodic states in the ergodic decomposition of
the thermal (Gibbs) state. Our examples include quantum spin systems and many-body boson systems. As a consequence, 
we elucidate the problem of equivalence between Bose-Einstein condensation and the quasi-average spontaneous 
symmetry breaking (SSB) discussed in \cite{LSYng}, \cite{LSYng1} for continuous boson systems. The multi-mode 
extended van den Berg-Lewis-Pul\'{e} condensation of type III \cite{vdBLP}, \cite{BZ} demonstrates that the only
physically reliable quantities are those that defined by Bogoliubov quasi-averages.
}
\end{abstract}

\newpage
\tableofcontents
\section{Introduction and summary}\label{sec:Intr-Summ}
The concept of Spontaneous Symmetry Breaking (SSB) is a central one in quantum physics, both in statistical
mechanics  and quantum field theory and particle physics. In this paper we restrict
ourselves to continuous SSB since the breaking of discrete symmetries has been extensively studied and it has
quite different properties, in particular regarding the Goldstone-Mermin-Wagner theorem for both zero $T=0$ and
non-zero $T>0$ temperatures, see e.g. \cite{SimonSM} and references given there.

The definition of SSB is also well-known since the middle sixties and is well
expounded in Ruelle's book \cite{Ru}, Ch.6.5.2., and references given there, as well as \cite{BR87}, Ch.4.3.4,
and, from the point of view of local quantum theory in \cite{Haag}, Ch.III.3.2. Roughly speaking, one starts from a
state (ground or thermal), assumed to be invariant under a symmetry group $G$, but which has a
nontrivial decomposition into extremal states, which may be physically interpreted as pure thermodynamic phases.
The latter, however, do not exhibit invariance under $G$, but only under a proper subgroup $H$ of $G$.

There are basically two ways of constructing extremal states: (1) by a choice of
boundary conditions (b.c) for  Hamiltonians $H_{\Lambda}$ in finite regions; (2) by replacing $H_{\Lambda} 
\rightarrow H_{\Lambda} + \lambda B_{\Lambda}$, where $B_{\Lambda}$ is a suitable extensive operator and $\lambda$ a
real parameter, taking first $\Lambda \nearrow \mathbf{Z}^{d}$ or $\Lambda \nearrow \mathbf{R}^{d}$, and then
$\lambda \to +0$ (or $\lambda \to -0$). Here one assumes that the states considered are locally normal or locally
finite, see e.g. \cite{Sewell1} and references there. Method (2) is known as Bogoliubov's \textit{quasi-averages}
method \cite{Bog07}-\cite{Bog70}.

Note that the method (1) is not of general applicability to, e.g., continuous many-body systems or quantum field 
theory. It is thus of particular interest to show that the Bogoliubov quasi-average "trick" may be shown to 
constitute a method, whose applicability is universal, explaining, at the same time, its physical meaning. 
This is one of the main purpose of our paper.

An important element of discussion is a general connection between SSB and Off-Diagonal Long-Range Order (ODLRO),
that was studying in papers by Fannes, Pul\`{e} and Verbeure \cite{FPV1} (see also \cite{PVZ}), by Lieb,
Seiringer and Yngvason (\cite{LSYng}, \cite{LSYng1}), and by S\"{u}t\"{o} \cite{Suto1}. The central role played by
ODLRO in the theory of phase transitions in quantum spin systems was scrutinised by Dyson, Lieb and
Simon \cite{DLS}, see also the review by Nachtergaele \cite{Ntg}. For its importance
in the theories of superconductivity and superfluidity, we refer to the books by Sewell \cite{Sewell} and
Verbeure \cite{Ver}, as well as to review \cite{SeW}, Sec.3.

As a consequence of our results, a general question posed by Lieb, Seiringer and Yngvason \cite{LSYng}
concerning the equivalence between Bose-Einstein condensation $\rm{(BEC)}_{qa}$ and  Gauge Symmetry Breaking
$\rm{(GSB)}_{qa}$ both defined via the one-mode Bogoliubov quasi-average is elucidated for any \textit{type} of
generalised BEC \textit{\`{a} la} van den Berg-Lewis-Pul\`{e} \cite{vdBLP} and \cite{BZ}.

\section{{Setup: continuous SSB, ODLRO and examples}} \label{sec:Setup}

To warm up we start by some indispensable basic notations and definitions, see e.g.,
\cite{BR87}, \cite{BR97}, \cite{Sewell1} and \cite{Wrepp}

Let $\mathcal{A}$ be a unital (i.e. $\mathds{1}\in {\cal A}$) quasi-local $C^*$-algebra of observables.
Recall that positive linear functionals $\omega$ over ${\cal A}$ are called \textbf{states} if they are normalised:
$\|\omega\| =1$. Note that these functionals are automatically continuous and bounded: 
$\|\omega\| = \omega(\mathds{1})$.
The state $\omega$ is called faithful if $\omega(A^* A) = 0$ implies $A = 0$.

To construct states and dynamics of quantum (boson) systems the $C^*$-setting is too restrictive and one
has to use the $W^*$-setting (\cite{BR87}, \cite{BR97}). One defines an abstract
$W^*$-algebra $\mathfrak{M}$ as a unital $C^*$-algebra that possesses (as a Banach space) a predual $\mathfrak{M}_*$,
i.e., $\mathfrak{M} = (\mathfrak{M}_*)^*$. Every abstract $W^*$-algebra is $\ast$-isomorphic to a \textit{concrete}
$W^*$-algebra $\mathcal{B}(\mathcal{H})$ of bounded operators on a Hilbert space $\mathcal{H}$. Now we can introduce
\textit{normal} states on $W^*$-algebra as those that any $\omega$ on the corresponding concrete $W^*$-algebra
is defined by a positive trace-class operator $\rho\in \mathcal{C}_{1}(\mathcal{H})$ with trace-norm 
$\|\rho\|_{1} = 1$ such that
\begin{equation*}
\omega (A) = {\rm{Tr}}_{{\cal H}} (\rho \  A) \ , \ {\rm{for \ all}} \ A \in \mathcal{B}(\mathcal{H}) \ .
\end{equation*}

For a finite system in the Hilbert space ${\cal H}_{\Lambda}$, the Gibbs (thermal) state is normal
\begin{equation}\label{2.1}
\omega_{\beta,\mu,\Lambda}(A) = {\rm{Tr}}_{{\cal H}_{\Lambda}} (\rho_{\Lambda} A) \ , \
{\rm{for \ all}} \ A \in \mathcal{B}(\mathcal{H}_{\Lambda}) \ .
\end{equation}
and defined by the trace-class density matrix
\begin{equation}\label{2.2}
\rho_{\Lambda} = \frac{\exp(-\beta(H_{\Lambda}-\mu N_{\Lambda}))}{\Xi_{\Lambda}(\mu,\beta)} \ .
\end{equation}
Here $\Lambda$ is a finite domain in $\mathbf{Z}^{d}$ for quantum spin systems, or in $\mathbf{R}^{d}$ for continuous
many-body systems and
$\Xi_{\Lambda}$ is the grand-canonical partition function
\begin{equation}\label{2.3}
\Xi_{\Lambda}(\mu,\beta) := {\rm{Tr}}_{{\cal H}_{\Lambda}} \exp(-\beta(H_{\Lambda}-\mu N_{\Lambda})) \ ,
\end{equation}
with $\beta = 1/k_{B}T$ the inverse temperature, $\mu$ the chemical potential for continuous quantum system.
For a boson continuous quantum system the Hilbert space ${\cal H}_{\Lambda}$ coincides with the symmetric Fock space
$\mathfrak{F}_{symm}(L^{2}(\Lambda))$, and ${\cal C}^{2}_{\Lambda} = \otimes_{i=1}^{N} {\cal C}_{i}^{2}$ for quantum
spin systems, with $N=V=|\Lambda|$, the number of points in $\Lambda$. The thermodynamic limit in both cases will
be denoted by $V \to \infty$.
Operator $A$ in (\ref{2.1}) is an element of a local algebra ${\mathfrak{M}}_{\Lambda}= {\cal B}({\cal H}_{\Lambda})$
of bounded operators on ${\cal H}_{\Lambda}$. By $H_{\Lambda}$ we denote the Hamiltonian of the system in
a finite domain $\Lambda$, and by $N_{\Lambda}$ the corresponding number operator. If
$\Omega_{\Lambda} \in {\cal H}_{\Lambda}$ is the ground-state vector of operator $H_{\Lambda} -\mu N_{\Lambda}$,
then the ground state ($T=0$) is defined by
\begin{equation}\label{2.4}
\omega_{\infty,\mu,\Lambda}(A) := (\Omega_{\Lambda}, A \Omega_{\Lambda}) \mbox{ for } A \in {\cal A}_{\Lambda} \ .
\end{equation}

By $\omega_{\beta,\mu}$ and $\omega_{\infty,\mu}$ we denote thermal and ground states for the infinite-volume
(thermodynamic) limit of the finite-volume states (\ref{2.1}) and (\ref{2.3}), in the sense that
\begin{equation}\label{2.5-6}
\omega_{\beta,\mu}(A) = \lim_{V \to \infty} \omega_{\beta,\mu,\Lambda}(A) \ \ {\rm{and}} \ \
\omega_{\infty,\mu}(A) = \lim_{V \to \infty} \omega_{\infty,\mu,\Lambda}(A) \ , \
A \in \bigcup_{\Lambda \subset (\mathbb{R}^d \, {\rm{or}}\,  \mathbb{Z}^d)}{\cal A}_{\Lambda} \ .
\end{equation}

\smallskip
Now we recall that a $C^*$-dynamics on a $C^*$-algebra ${\cal A}$ is a strongly continuous one-parameter
group of $\ast$-automorphisms: $\mathbf{R}\ni t \mapsto \tau_{t}$ of $C^*$.
Then a $C^*$-\textit{dynamical system} is a corresponding pair $({\cal A}, \tau_{t})$. Note that
the strong continuity of $\{\tau_{t}\}_{t\in \mathbb{R}}$ on ${\cal A}$ means that the map
$t \mapsto \tau_{t}(A)$ is norm-continuous for any $A\in {\cal A}$. Therefore, $C^*$-\textit{dynamical systems}
are completely characterised by the corresponding densely defined and closed in ${\cal A}$ infinitesimal
generators.

It is also well-known that the $C^*$-{dynamical systems} are too restrictive for boson systems, that forces
to use the $W^*$-setting. Let $\mathfrak{M}$ be a von Neumann algebra ($W^*$-algebra) and let
$\mathbf{R}\ni t \mapsto \tau_{t}$ be a one-parameter group of weak*-continuous $\ast$-automorphisms
($W^*$-dynamics) of $\mathfrak{M}$. Then the pair $(\mathfrak{M}, \tau_{t})$ is called a
$W^{\ast}$-\textit{dynamical system}.
The continuity condition on the group $\{\tau_{t}\}_{t\in \mathbb{R}}$ means that
the weak*-densely defined and closed in $\mathfrak{M}$ infinitesimal generator corresponding to
the $W^*$-dynamics can be defined in the weak*-topology similar to the $C^*$-setting.

We comment that it is this $W^*$-setting, which is appropriate for representations of the Canonical Commutation
Relations (CCR) and description of boson systems by the Weyl algebra  \cite{BR97}, \cite{PiMe}. We shall also
use it for quantum spin systems, and, therefore, throughout the whole paper.
\begin{defi} \label{QDS}
Consider a $W^{\ast}$-{dynamical system} $(\mathfrak{M}, \tau_{t})$. A state on $\mathfrak{M}$ is
called $\tau$-invariant if $\omega \circ \tau_{t} = \omega$ for all $t\in \mathbb{R}$. If in addition this
state is normal, we refer to the triplet $(\mathfrak{M}, \tau_{t}, \omega)$ as to a Quantum Dynamical System (QDS)
generated by $(\mathfrak{M}, \tau_{t})$.
\end{defi}
Recall that GNS representation $\pi_{\omega}$ of the QDS, which is induced by the invariant state $\omega$, is
denoted by the triplet $({\cal H}_{\omega}, \pi_{\omega}, \Omega_{\omega})$. Here, $\Omega_{\omega}$ is a cyclic
vector for $\pi_{\omega}(\mathfrak{M})$ in the Hilbert space ${\cal H}_{\omega}$.
The unicity of the GNS representation implies that there exists a unique one-parameter
group $t \mapsto U_{\omega}(t)$ of unitary operators on ${\cal H}_{\omega}$ such that
\begin{equation} \label{GNS}
\pi_{\omega}(\tau_{t}(A)) = U_{\omega}(t)\pi_{\omega}(A)U_{\omega}^{*}(t) \ ,
\ U_{\omega}(t) \Omega_{\omega} = \Omega_{\omega} \ ,
\end{equation}
for any $t\in \mathbb{R}$ and $A \in \mathfrak{M}$.

Since we assumed $\omega$ to be normal, the group $\{U_{\omega}(t)\}_{t\in \mathbb{R}}$ is
strongly continuous and there exists (by the Stone theorem) a unique self-adjoint generator $H_{\omega}$
of this unitary group such that
\begin{equation} \label{GNS-H}
\pi_{\omega}(\tau_{t}(A)) = e^{i t H_{\omega}}\pi_{\omega}(A)e^{- i t H_{\omega}} \ , \
H_{\omega} \Omega_{\omega} = 0  \ .
\end{equation}

We note that GNS construction applied directly to a $C^*$-{dynamical system} with invariant state $\omega$
defines a normal extension of this state to an enveloping von Neumann algebra. Therefore, it maps the
$C^*$-{dynamical system} into a $W^*$-{dynamical system} with a normal invariant state. Hence, instead of
QDS one can start with GNS representation of the $C^*$-algebra ${\cal A}$.

{In the context of infinite boson system we suppose also that the time-invariant $\omega$ is
such that restriction to ${\cal A}_{\Lambda}$ (or $\mathfrak{M}_{\Lambda}$) is given by
$\omega_{\beta,\mu,\Lambda}$ (\ref{2.1}),(\ref{2.2}), (\ref{2.3})}

{{Now let $G$ be a group and $\{\tau_{g}\}_{g \in G}$ be the associated group of $\ast$-automorphisms
in ${\cal A}$.  Suppose that $\tau_{g}$ leaves $\omega$ invariant:
\begin{equation}\label{2.9}
\omega(\tau_{g}(A)) = \omega(A), \ \forall A \in {\cal A}, \  \forall g \in G \ .
\end{equation}
Then one can find on the GNS Hilbert space ${\cal H}_{\omega}$ a unique group of unitary operators
$\{U_{g}\}_{g \in G}$ such that
\begin{equation}\label{2.10}
\pi_{\omega}(\tau_{g}(A))= U_{g} \pi_{\omega} U_{g}^{\ast} \ \mbox{ with } \ U_{g}\Omega_{\omega}=\Omega_{\omega} \ .
\end{equation}
It is easy to show (see, e.g., \cite{Wrepp}) that the natural candidate for $U_{g}$, given by
\begin{equation}\label{2.11}
U_{g} \pi_{\omega}(A) \Omega_{\omega}= \pi_{\omega}(\tau_{g}(A))\Omega_{\omega} \ ,
\end{equation}
indeed fulfills these requirements. The $G$-invariant states forms a convex and compact in the weak*-topology set,
that we denote by $E_{{\cal A}}^{G}$.

The same properties are evidently shared by the set $E_{{\cal A}}$ of all states on ${\cal A}$.
An \textbf{extremal invariant} or \textbf{ergodic} state is a state $\omega \in E_{{\cal A}}^{G}$, which cannot be
written as a proper convex combination of two distinct states $\omega_{1},\omega_{2} \in E_{{\cal A}}^{G}$:
\begin{equation}\label{2.12}
\omega \ne \lambda \omega_{1} + (1-\lambda) \omega_{2} \mbox{ with } 0<\lambda<1 \ \mbox{ unless } \
\omega_{1}=\omega_{2}=\omega \ .
\end{equation}
}}
{There exists an alternative characterization:} we say that a state $\omega_{1}$ \textbf{majorizes}
another state $\omega_{2}$ if $\omega_{1}-\omega_{2}$ is a positive linear functional on ${\cal A}$, i.e.,
$(\omega_{1}-\omega_{2})(A^* A) \ge 0 \ \ \forall A \in {\cal A}$. Clearly, if a state is a convex combination of
two others, it majorizes both, and a state $\omega$ is said to be \textbf{pure} if the only positive linear
functionals majorized by $\omega$ are of the form $\lambda \omega$, with $0\le \lambda \le 1$. By \cite{BR87},
Theorem 2.3.15, we are allowed to use the terms pure and extremal interchangeably. When (\ref{2.12}) does not hold,
it is natural to consider $\omega$ as a \textbf{mixture} of two \textbf{pure phases} $\omega_{1}$ and $\omega_{2}$,
with proportions $\lambda$ and $(1-\lambda)$, respectively.

Thermal states $\omega_{\beta,\mu}$ satisfy the equilibrium (KMS) condition (\cite{BR97},\cite{Hug}) and will
be called KMS or thermal equilibrium states, or, for short, thermal states. The \textbf{commutant}
$\pi_{\omega}({\cal A})^{'}$ of $ \pi_{\omega}({\cal A})$ is defined as $\pi_{\omega}({\cal A})^{'} =
\{B \in  {\cal B}({\cal H}_{\omega}): \  [A,B]=0 \ \ \forall A \in \pi_{\omega}({\cal A})\}$. the strong closure
of $\pi_{\omega}({\cal A})$, called the von neumann algebra generated by $ \pi_{\omega}({\cal A})$, which also
equals $ \pi_{\omega}({\cal A})^{''}$ by von Neumann's theorem \cite{BR87}, is called a \textbf{factor} if
its \textbf{center}
\begin{equation}\label{2.13}
Z_{\omega} = \pi_{\omega}({\cal A})^{'} \cap \pi_{\omega}({\cal A})^{''} \ ,
\end{equation}
is a multiple of the identity operator
\begin{equation}\label{2.14}
Z_{\omega} = \{ \mathbb{C} \  \mathds{1}\} \ .
\end{equation}
The corresponding representation is called \textbf{factor} or \textbf{primary}, and the extension of $\omega$ to
$\pi_{\omega}({\cal A})^{''}$ is called a factor or primary state. Consider the central decomposition of a KMS state
$\omega_{\beta}$  \cite{BR87} (we omit the $\mu$ for brevity):
\begin{equation}\label{2.15}
\omega_{\beta}(A) = \int_{E_{\cal A}^{G}} d\mu(\omega_{\beta}^{'}) \omega_{\beta}^{'}(A)\ ,
\end{equation}
which, for a KMS state is identical to the extremal or ergodic decomposition,
see Theorem 4.2.10 of \cite{BR87}. The states $\omega_{\beta}^{'}$ in (\ref{2.15}) are extremal or factor states,
and the decomposition is along the center $ Z_{\omega_{\beta}}$ which is of the form (\ref{2.14}). In the examples
we shall treat, $Z_{\omega_{\beta}}$ coincides with the so-called \textbf{algebra at infinity} (\cite{BR87},
Example 4.2.11). Let $\omega$ be a spatially ($\mathbf{Z}^{d}$ - or $\mathbf{R}^{d}$) - translation invariant
state (we shall no longer distinguish these two possibilities explicitly):
\begin{equation}\label{2.16}
\omega(\tau_{{x}}(A)) = \omega(A) \  \ \forall  A \in {\cal A} , \ \forall {x} \ ,
\end{equation}
where $\tau_{{x}}$ denotes the group of automorphisms of ${\cal A}$ corresponding to translations. Let us define
\begin{equation}\label{2.17}
\eta(A) := s-\lim_{V \to \infty} \eta_{\Lambda}(A) \ ,
\end{equation}
where
\begin{equation}\label{2.18}
\eta_{\Lambda}(A) = \frac{1}{V}\int_{\Lambda} d{x} \pi_{\omega}(\tau_{{x}}(A)) \ ,
\end{equation}
again not distinguishing the lattice from the continuous case, in the former one has a sum instead of the
integral in (\ref{2.18}). The existence of (\ref{2.17}) is well-known, see \cite{BR97}, or Proposition 6.7 in
\cite{MWB}. Then by construction $\eta(A) \in Z_{\omega}$. If $\omega$ is an extremal (=factor=primary), which is
also ergodic for space translations, then (\ref{2.14}) holds and therefore
\begin{equation}\label{2.19}
\eta(A) = \omega(A) \ \mathds{1} \ .
\end{equation}
Hence, the states occurring in the extremal or ergodic decomposition of a KMS state correspond to ''freezing''
the observables at infinity to their expectation values. In correspondence with (\ref{2.17}), we extend
(\ref{2.5-6}) to space averages by
\begin{equation}\label{eqn3.2.20}
\omega_{\beta,\mu}(\prod_{i=1}^{m} \eta(A_{i}) B) = \\
\lim_{V \to \infty} \omega_{\beta,\mu,\Lambda}(\prod_{i=1}^{m}\eta_{\Lambda}(A_{i}) B)
=\lim_{V \to \infty} \omega_{\beta,\mu}(\prod_{i=1}^{m} \eta_{\Lambda})(A_{i}) B)\ .
\end{equation}
Here we assumed that $\omega_{\beta,\mu,\Lambda} $ is space translation-invariant, which may be achieved
by imposing the periodic b.c. on $\Lambda$.

Let, now, $G$ be a group, $\{\tau_{g}\}$ denote the corresponding group of $\ast$-automorphisms of ${\cal A}$, and
assume that $\tau_{g} \circ \tau_{{x}}= \tau_{{x}} \circ \tau_{g}$ for all $g \in G$ and ${x}$, i.e.,
$G$ commutes with space translations. We assume henceforth that all states are space translation-invariant,
(\ref{2.16}), and thus all states in decomposition (\ref{2.15}) are also invariant under space translations.
\begin{defi} \label{SSB}
We say that the state $\omega$ undergoes a (\textit{conventional}) Spontaneous Symmetry Breaking (SSB) of the
group $G$ if:\\
(i) $\omega$ is $G$-invariant, i.e., (\ref{2.9})-(\ref{2.11}) hold;\\
(ii) $\omega$ has a nontrivial decomposition (\ref{2.15}) into ergodic states $\omega_{\beta}^{'}$, which means
that at least two such distinct states occur in representation (\ref{2.15}), and
\begin{equation}\label{2.21}
\omega_{\beta}^{'}(\tau_{g}(A)
\ne \omega_{\beta}^{'}(A) \ ,
\end{equation}
for some $g \in G$, and for some $A \in {\cal A}$.
\end{defi}
As previously remarked, we shall use the W* description: ${\cal A}$ shall henceforth be identified with the
von neumann algebra $\pi_{\omega}({\cal A})^{''}$ corresponding to a given state $\omega$; thus, $\eta(A)$ is a
special element of ${\cal A}$.

Note that if there is no nontrivial decomposition, then there exists only one equilibrium state, which is then
automatically $G$-invariant.

The physical interpretation of condition (ii) in Definition \ref{SSB} is well-known (see \cite{Sewell1}, or 
\cite{Ru}). By (\ref{2.19}), for an ergodic state $\omega^{'}$ in (\ref{2.15}),
\begin{equation}
\lim_{V \to \infty} \frac{1}{V} \, \int_{\Lambda}d{x}\ \omega^{'}(\tau_{{x}}(A) B)
= \omega^{'}(A) \, \omega^{'}(B) \ \  \ \forall A,B \in {\cal A} \ .
\label{eqn2.2.22}
\end{equation}
By (\ref{2.14}) and (\ref{eqn2.2.22}), the spatial averages do not fluctuate in an ergodic state $\omega^{'}$:
\begin{equation}
\lim_{V \to \infty} \omega^{'}\left\{\left(\frac{1}{V} \,\int_{\Lambda}d{x} \ \tau_{{x}}(A)\right)^{2}-
\left(\frac{1}{V} \, \int_{\Lambda} d{x}\ \omega^{'}(\tau_{{x}}(A))\right)^{2}\right\} = 0 \ .
\label{eqn2.2.23}
\end{equation}
This is a characteristic property of a \textbf{pure} thermodynamic phase, in which average values, such as the
density, do not fluctuate (in contrast to a mixture).

How does this relate to SSB ? The part (ii) of Definition \ref{SSB} implies that $\tau_{g}$ cannot be implemented
by a group of unitary operators in ${\cal H}_{\omega}$ in the form (\ref{2.11}), in particular suitable generators
of the unitary group do not exist. A natural alternative to (\ref{eqn2.2.22}) is to replace it
(see \cite{SwiecaJ}, \cite{Wrepp}) by
\begin{equation}\label{2.24}
\lim_{R \to \infty, \delta \to 0} \omega_{\beta}^{'} ([Q_{R,\delta},A]) \ne 0 \mbox{ for some } A \in {\cal A}_{L}
\ .
\end{equation}
Here ${\cal A}_{L}$ is the dense subalgebra of local observables, and $Q_{R,\delta}$ is a smooth approximation
to the charge in space and time, i.e.,
\begin{equation}\label{2.25}
Q_{R,\delta} := \int d{x} dt f_{R}({x}) f_{d}(t) j^{0}({x},t) \ ,
\end{equation}
with $\lim_{|{x}| \to \infty} f_{R}({x})= 1$, $ f_{\delta}$ tends to delta-function as $\delta \to 0$,
and $ j^{0}({x},t)$ is the ''charge density''. In statistical mechanics one may ignore time-smoothing,
and choose $f_{R}$ as characteristic function of a region $\Lambda$. The limit (\ref{2.24}) exists as a consequence
of locality \cite{SwiecaJ}. For quantum statistical mechanics one uses the the property of "causality"
$[{\cal A}_{\Lambda}, {\cal A}_{\Lambda^{'}}]= 0$ if $\Lambda \cap \Lambda^{'} = \emptyset$.

To illustrate the ideas presented above we recall a standard example of quantum spin systems corresponding to
the simplest Heisenberg ferromagnet
\begin{equation}\label{2.26}
H_{\Lambda} = - \sum_{{x},{y} \in \Lambda; \|{x}-{y}\|=1} {\sigma}_{{x}} \cdot
{\sigma}_{{y}} \ ,
\end{equation}
where $\sigma_{{x}}^{i},i=1,2,3$ are the Pauli matrices at ${x}$, on the Hilbert space
${\cal H}_{\Lambda} = \otimes_{{x} \in \Lambda} {\cal C}_{{x}}^{2}$. Assuming that $H_{\Lambda}$ in
(\ref{2.26}) is defined with periodic b.c., so that the momentum is also well-defined, the Gibbs state
$\omega_{\beta,\Lambda}$ in (\ref{2.1}) (with $\mu = 0$) is invariant under the rotation group $G= SO(3)$.
Hence, $\omega_{\beta}$ satisfies (\ref{2.9}) with $G=SO(3)$, and, moreover, (\ref{2.16}) also holds by
translation invariance of $\omega_{\beta,\Lambda}$.

The "charge" (\ref{2.25}) coincides with magnetisation
\begin{equation}\label{2.27}
{Q}_{\Lambda} = {M}_{\Lambda} = \sum_{{x} \in \Lambda} {\sigma}_{{x}} \ .
\end{equation}
In an ergodic state the spatial average (\ref{2.17}), (\ref{2.18}) of the observable ${\sigma}$,
\begin{equation}\label{2.28}
\eta({\sigma}) = s-\lim_{V \to \infty} (\eta_{\Lambda} =\frac{1}{V} \sum_{{x} \in \Lambda} {\sigma}_{{x}}) \ ,
\end{equation}
is equal by (\ref{2.19}) to
\begin{equation}\label{2.29}
\eta({\sigma}) =  \lambda \, {n} \ ,
\end{equation}
where ${n}$ is a fixed unit vector and coefficient $\lambda = \lambda(\beta,\mu)$. Note that a rotation
$g=R \in G$ acts on ${n}$, by (\ref{2.28}),(\ref{2.29}), in the form
\begin{equation}\label{2.30}
\tau_{g}(\eta({\sigma})) = \lambda \, R \, {n} \ .
\end{equation}
Since (\ref{2.15}) is a central decomposition, for the Gibbs state we may write it in the form
\begin{equation}\label{2.31}
\omega_{\beta}(A) = \int d\mu_{{n}} \omega_{\beta,{n}}(A) \ ,
\end{equation}
where $\mu$ is the normalized measure on the sphere $S_{2}$ and each $\omega_{\beta,{n}}$ is ergodic. Further,
\begin{equation}
\omega_{\beta,{n}}(\tau_{g}(A)) = \omega_{\beta,{n}}(\tau_{g}(\eta(A)))
= \omega_{\beta,R {n}}(\eta(A)) \ .
\label{eqn2.2.32}
\end{equation}

Now we recall the concept of \textit{Off-Diagonal Long Range Orde}r (ODLRO), which is relevant to our discussion 
of SSB in the Heisenberg ferromagnet: in this definition $\omega_{\beta}$ is assumed to be an equilibrium state of 
the Heisenberg ferromagnet (2.25): for a discussion of other examples, see remark 2.5.
\begin{defi}\label{ODLRO}
For a given $\beta$ the state $\omega_{\beta}$ is said to exhibit ODLRO if
\begin{equation}\label{2.33}
\lim_{V \to \infty} \omega_{\beta} (\eta_{\Lambda}({\sigma})^{2}) > 0 \ .
\end{equation}
\end{defi}
Since $\eta_{\Lambda}$ is given  by (\ref{2.18}), the both $\sigma$ and $\eta_{\Lambda}(\sigma)$ are three-component
vectors, and $(\eta_{\Lambda}(\sigma))^{2} = \|\eta_{\Lambda}(\sigma)\|^2$.
Hence, the space-averaged magnetization: $\eta$ (\ref{2.28}),  \textbf{fluctuates} in the state $\omega_{\beta}$.
The following well-known  proposition relates ODLRO and SSB for the Heisenberg ferromagnet:
\begin{proposition}\label{prop:2.1}
If $\omega_{\beta}$ exhibits conventional ODLRO, it undergoes the SSB defined by (\ref{2.21}), with $A$ defined by
(2.27). Conversely, if (2.20) holds for some $\omega_{\beta,n}$ in the decomposition (2.30), with $A$ given
by (2.27), then (2.32) holds.
\end{proposition}
\begin{proof} {If $\omega_{\beta}$ exhibits ODLRO, it follows from (2.21) and (2.28) that $\lambda \ne 0$ in (2.28),
and thus the ergodic decomposition (2.30) is nontrivial. hence, SSB holds, with $A$ in (2.20) given by $\eta(A)$,
defined by (2.27). The converse statement is a direct consequence of the ergodicity of$\omega_{\beta,n}$, and the
fact that (2.20) implies that $\lambda \ne 0$.}
\end{proof}

\begin{rem}\label{rem:2.1}
The ergodic states are not invariant under $G=SO(3)$ but rather under the isotropy
(stationary) subgroup $H_{{n_{0}}}$ of $G$, and $S_{d-1}$ may be identified as the harmonic space $G/H$.
\end{rem}
\begin{rem}\label{rem:2.2}
The connection between ODLRO and the existence of several equilibrium states for quantum
spin systems was first pointed out by Dyson, Lieb and Simon in their seminal paper \cite{DLS}, see also the
review by Nachtergaele \cite{Ntg} and references given there. By \cite{DLS}, both the spin one-half XY model
for $\beta \ge \beta_{c}^{1}$, and the Heisenberg antiferromagnet for suitable spin and
$\beta \ge \beta_{c}^{2}$, with $\beta_{c}^{1},\beta_{c}^{2}$ explicitly given in \cite{DLS}, display ODLRO
in the sense of definition 2.3, but with different $\eta_{\Lambda}$ in (2.27) (for the antiferromagnet the sum
over $x \in \Lambda $ being replaced by a sum over $\Lambda \cap A$, where $\mathbf{Z}^{d} = A \cup B$, $A$ and
$B$ being disjoint sublattices. We expect that proposition 2.1 is applicable to the above mentioned cases,
yielding SSB (of the rotation group ($SO(2)$ in the XY case) according to Definition \ref{SSB}, but a choice of
$\eta_{\Lambda}$ for a general Heisenberg hamiltonian, with arbitrary spin, is not known, as well as what the
correct order parameters are, and how the set of pure phases should be parametrized and constructed (We thank
B. Nachtergaele for this last remark). We therefore restrict ourselves to the the ferromagnet as our quantum spin
example.
\end{rem}
\begin{rem}\label{rem:2.3}
By (\ref{2.29}) we have different values for the "charge density" $\eta({\sigma})$
labelled by ${n} \in S_{2}$. By a well-known result (see, e.g., \cite{MWB}, Corollary 6.3), the GNS
representations $\pi_{\omega_{{n}}}$ associated to the corresponding states $\omega_{{n}}$ in the
(central) decomposition (\ref{2.31}) are not unitary equivalent (they are, more precisely, disjoint, see
Definition 6.6 in \cite{MWB}), and the GNS Hilbert space splits into a direct integral of disjoint "sectors"
${\cal H}_{{n}}$ (see e.g. \cite{BR87}).
\end{rem}
{{We note that in this respect the case of boson systems is more complicated than spin lattice systems.
It becomes clear even on the level of the perfect Bose-gas.

To see this, consider the Perfect Bose-gas (PBG) in a three-dimensional anisotropic parallelepiped
$\Lambda:= V^{\alpha_1}\times V^{\alpha_2}\times V^{\alpha_3}$, with \textit{periodic boundary
condition} (p.b.c.) and $\alpha_1 \geq \alpha_2 \geq \alpha_3$, $\alpha_1 + \alpha_2 + \alpha_3 = 1$, i.e.
the volume $|\Lambda| = V$.  In the
boson Fock space $\mathcal{F}:= \mathcal{F}_{boson}(\mathcal{L}^2 (\Lambda))$  the Hamiltonian of this system
for the grand-canonical ensemble with chemical potential $\mu < 0$ is defined by :
\begin{eqnarray}\label{G-C-PBG}
H^{0}_{\Lambda} (\mu) \,  = T_{\Lambda} - \mu \, N_{\Lambda} =
\sum_{k \in \Lambda^{*}} (\varepsilon_{k} - \mu)\, b^{*}_{k} b_{k} \ .
\end{eqnarray}
Here one-particle kinetic-energy operator spectrum $\{\varepsilon_{k} = k^2\}_{k \in \Lambda^{*}}$, where the dual to
$\Lambda$ set is :
\begin{equation}\label{dual-Lambda}
\Lambda ^{\ast }= \{k_{j}= \frac{2\pi }{V^{{\alpha_{j}}}}n_{j} : n_{j}
\in \mathbb{Z} \}_{j=1}^{d=3}
\ \ \ {\rm{then}} \ \ \ \varepsilon _{k}= \sum_{j=1}^{d} {k_{j}^2} \ .
\end{equation}
We denote by $b_{k}:=b(\phi_k^\Lambda)$ and $b^{*}_{k}= b^{*}(\phi_k^\Lambda)$ the
boson annihilation and creation operators in the Fock space $\mathcal{F}$. They are indexed by the ortho-normal
basis $\{\phi_k^\Lambda(x) = e^{i k x}/\sqrt{V}\}_{k \in \Lambda^{*}} \subset \mathcal{L}^2 (\Lambda)$ generated by
the eigenfunctions of the self-adjoint one-particle kinetic-energy operator $(- \Delta)_{p.b.c.}$ in
$\mathcal{L}^2 (\Lambda)$. Formally these operators satisfy the Canonical Commutation Relations (CCR):
$[b_{k},b^{*}_{k'}]=\delta_{k,k'}$. Then $N_k =  b^{*}_{k} b_{k}$ is occupation-number operator of the one-particle
state $\phi_k^\Lambda$ and $N_{\Lambda} = \sum_{k \in \Lambda^{*}} N_k$ denotes the total-number operator in 
$\Lambda$.

If we denote by $\omega_{\beta,\mu,\Lambda}^{0}(\cdot)$ the grand-canonical Gibbs state of the PBG  generated by
(\ref{G-C-PBG}), then the problem of existence of conventional Bose-Einstein condensation is related to solution
of the equation
\begin{equation}\label{BEC-eq}
\rho = \frac{1}{V} \sum_{k \in \Lambda^{*}} \omega_{\beta,\mu,\Lambda}^{0}(N_k) =
\frac{1}{V} \sum_{k\in \Lambda ^{\ast}}\frac{1}{e^{\beta \left(\varepsilon_{k}-\mu \right)}-1} \ ,
\end{equation}
for a given total particle density  $\rho$ in $\Lambda$. Note that by (\ref{dual-Lambda}) the thermodynamic limit
$\Lambda \uparrow \mathbb{R}^3$ in the right-hand side of (\ref{BEC-eq})
\begin{equation}\label{I}
\mathcal{I}(\beta,\mu) = \lim_{\Lambda} \frac{1}{V} \sum_{k \in \Lambda^{*}} \omega_{\beta,\mu,\Lambda}^{0}(N_k)
= \frac{1}{(2\pi)^3}\int_{\mathbb{R}^3} d^3 k \ \frac{1}{e^{\beta \left(\varepsilon_{k}-\mu \right)}-1} \ ,
\end{equation}
exists for any $\mu <0$. It reaches its (finite) maximal value $\mathcal{I}(\beta,\mu =0) = \rho_c(\beta)$, which is
called the critical particle density for a given temperature.

The existence of finite $\rho_c(\beta)$ triggers (via \textit{saturation mechanism}) a non-zero BEC
$\rho_0(\beta) := \rho - \rho_c(\beta)$, when the total particle density $\rho > \rho_c(\beta)$.

Note that for $\alpha_1 < 1/2$, the whole condensate is sitting in the one-particle ground state mode $k=0$:
\begin{eqnarray*}
&&{\rho_0} (\beta)= {\rho} - \rho _{c}(\beta) = \lim_{\Lambda} \frac{1}{V} \omega_{\beta,\mu,\Lambda}^{0}(N_0)
= \lim_{\Lambda} \frac{1}{V} \ \left\{e^{-\beta \, {\mu_{\Lambda}(\beta,\rho\geq
\rho _{c}(\beta))} }-1\right\}^{-1}\\
&&{\mu_{\Lambda}(\beta,\rho\geq \rho _{c}(\beta))}=  {- \, \frac{1}{V}} \ \frac{1}
{\beta(\rho -\rho _{c}(\beta))} + {o}({1}/{V}) \ ,
\end{eqnarray*}
where $\mu_{\Lambda}(\beta,\rho)$ is a unique solution of equation (\ref{BEC-eq}).

This is a well-known \textit{conventional} (or the \textit{type} I \cite{vdBLP}) condensation. In particular, in
this case it make sense the ODLRO for the Bose-field
\begin{equation}\label{b-field}
b(x) = \sum_{k \in \Lambda^{*}} b_{k} \phi_{k}^{\Lambda}(x) \ .
\end{equation}
Indeed, by Definition \ref{ODLRO} one gets for the spacial average of (\ref{b-field})
\begin{equation}\label{PBG-ODLRO}
\lim_{\Lambda} \omega_{\beta,\mu,\Lambda}^{0}(\frac{1}{V}\int_{\Lambda}dx b^*(x)\ \frac{1}{V}\int_{\Lambda}dx b(x))=
\lim_{\Lambda} \omega_{\beta,\mu,\Lambda}^{0}(\frac{b^{*}_{0} b_{0}}{V}) = \rho_{0} (\beta) \ ,
\end{equation}
i.e. the ODLRO coincides with the condensate density \cite{Ver}.

For $\alpha_1 = 1/2$ (the Casimir box \cite{ZBru})
one observes the infinitely-many levels macroscopic occupation called the \textit{type} II condensation.

On the other hand, when $\alpha_1 > 1/2$ (van den Berg-Lewis-Pul\'{e} boxe \cite{vdBLP}) one obtains
\begin{equation}\label{BEC=0}
\lim_{\Lambda} \omega_{\beta,\mu,\Lambda}^{0}(\frac{b^{*}_{k} b_{k}}{V}) =
\lim_{\Lambda}\frac{1}{V}\left\{e^{\beta(\varepsilon_{k}-{\mu_{\Lambda}(\beta,\rho)})}-
1 \right\}^{-1} = 0  \ , \  \forall k \in \Lambda^{*} \ ,
\end{equation}
i.e., there is no macroscopic occupation of any mode for any value of particle density $\rho$. But a
generalised BEC (gBEC of type III) does exist in the following sense:
\begin{equation}\label{gBEC}
\rho -\rho_{c}(\beta)= \lim_{\eta \rightarrow +0}\lim_{\Lambda }
\frac{1}{V}\sum_{\left\{ k\in \Lambda^{\ast }, \left\| k\right\|
\leq \eta \right\}}\left\{e^{\beta(\varepsilon_{k}- {\mu_{\Lambda}(\beta,\rho)})}- 1
\right\}^{-1}  , \ {\rm{for}} \ \ \rho > \rho_{c}(\beta) \ .
\end{equation}
Note that (\ref{PBG-ODLRO}) and (\ref{BEC=0}) imply triviality of the ODLRO, whereas the condensation
in the sense (\ref{gBEC}) is nontrivial.

We comment that this unusual condensation is not exclusively due to the special geometry $\alpha_1 > 1/2$.
In fact the same phenomenon of the gBEC (\textit{type} III) \cite{BZ} happens due to interaction in the model with
Hamiltonian \cite{ZBru}:
\begin{equation}\label{Int-TypeIII}
H_{\Lambda }= {\sum_{k\in \Lambda^{*}} }\varepsilon_{k}b_{k}^{*}b_{k}+
\frac{a}{2V}{\sum_{k\in\Lambda^{*}}} b_{k}^{*}b_{k}^{*}b_{k}b_{k}\ , \ \text{ } a>0 \ .
\end{equation}

These examples show that connection between BEC, ODLRO, and SSB is a subtle matter. This motivates and bolsters
a relevance of the Bogoliubov \textit{quasi-average method} \cite{Bog07}-\cite{Bog70}, that we discuss in the
next two sections.}}

\section{{Selection of pure states by the Bogoliubov quasi-averages: spin systems}} \label{sec:QA-spin}
Considering further the simple example of spin system (\ref{2.26}) for the sake of argument,
at least two methods of selecting pure states may be suggested: (1) by taking in (\ref{2.1}), (\ref{2.2})
$H_{\Lambda}$
with special boundary conditions (b.c.), i.e., upon imposing on the boundary $\partial \Lambda$ of $\Lambda$
\begin{equation}\label{3.1}
|{n})_{{x}} \mbox{ such that } {\sigma}_{{x}} |{n})_{{x}}= |{n})_{{x}}
\end{equation}
The above choice leads, presumably, to the limiting states $\omega_{\beta,{n}}$ in (\ref{2.31}); (2) by replacing in
(\ref{2.1}), (\ref{2.2}) $H_{\Lambda}$ by the \textit{quasi-}Hamiltonian
\begin{equation}\label{3.2}
H_{\Lambda,{B}} := H_{\Lambda} + H_{\Lambda}^{{B}} \ ,
\end{equation}
with the \textit{symmetry-breaking} vector field ${B} \, n$ directed along the unit vector $n$:
\begin{equation}\label{3.3}
H_{\Lambda}^{{B}} = - {B} \, n \cdot \sum_{{x} \in \Lambda} {\sigma}_{{x}} \ , \ B > 0 \ .
\end{equation}
We take $B \to +0$ after the thermodynamic limit $V \to \infty$. This method, which is known as the
{Bogoliubov \textit{quasi-averages}} (\cite{Bog07}-\cite{Bog70}, \cite{ZBru} ), is currently employed as a trick,
i.e., without explicit connection to ergodic states. The quantity $\sum_{{x} \in \Lambda} {\sigma}_{{x}}$
(the magnetization) in the symmetry-breaking field is known as the \textbf{order parameter}. As spelled out in
(\ref{3.3}), it is appropriate to the Heisenberg ferromagnet (\ref{2.26}) and for the XY model, but not
for the antiferromagnet, in which case the order parameter should be replaced by the sub-lattice magnetization
$\sum_{{x} \in \Lambda \cap A} {\sigma}_{{x}}$, where $\mathbf{Z}^{d}= A \cup B$, $A,B$ denoting
two disjoint sublattices.

If we consider first $0 < \beta < \infty$, $G=SO(3)$ and $H_{\Lambda}$ the Hamiltonian (\ref{2.26}) (or
its antiferromagnetic or XY analog), with free or periodic b.c., then $H_{\Lambda}$ is G-invariant, and
thus $\omega_{\beta,\Lambda}$, defined by (\ref{2.1}),(\ref{2.2}), is also G-invariant. Taking, now, $H_{\Lambda}$
with the b.c. (\ref{3.1}), \textbf{both} $H_{\Lambda}$ and $\omega_{\beta,\Lambda}$ are \textbf{not} G-invariant.
Consider, now, $\beta = \infty$, i.e., theground state, with $H_{\Lambda}$ given by (\ref{2.26}),
defined with free or periodic b.c.. Again, $H_{\Lambda}$ is
invariant under $G$, and we may regard a ground state
\begin{equation}\label{3.5}
\omega_{\infty,\Lambda} = (\Omega_{\Lambda}, \cdot \Omega_{\Lambda})) \ ,
\end{equation}
with
\begin{equation}\label{3.6}
|\Omega_{\Lambda} = \otimes_{{x} \in \Lambda} |{n})_{{x}} \ .
\end{equation}
Then, clearly, $\omega_{\infty,\Lambda}$ as well as its infinite volume counterpart is \textbf{not} G-invariant.
Note that (\ref{3.5}) leads, however, presumably to the ergodic states $\omega_{\infty,{n}}$ in the decomposition
(\ref{2.31}), when taking the weak* limit as $\Lambda  \nearrow \mathbf{Z}^{3}$.

If we take, however, the weak* limit, as $\beta \to \infty$ along a subsequence, of $\omega_{\beta}$, it may
be conjectured that the $G$-invariant ground state
\begin{equation*}
\omega_{\infty} := \int d\mu_{{n}} \omega_{\infty,{n}} \ ,
\end{equation*}
is obtained. The limits $V \to \infty$ and $\beta \to \infty$ are not expected to commute, and we believe,
in consonance with the third principle of thermodynamics \cite{WreA}, that it is more adequate, both
physically and mathematically, to regard the states $\omega_{\beta}$ for $0 < \beta < \infty$ as fundamental,
with ground states defined as their (weak*) limit as $\beta \to \infty$ (along a subsequence or subnet).
In this sense, the assertion found in most textbooks, see also \cite{LSYng} beginning of
Section 2, that SSB occurs when the Hamiltonian is invariant, but not the state, is not correct, or, at least,
not precise. Note, however, that, in the textbooks, "state" is understood as the ground state or the vacuum
state, but not as the thermal state, for which the equivalence between the invariance of the Hamiltonian and
the state is essentially obvious.

If one uses the method of Bogoliubov quasi-averages, such difficulties do not appear, because
$\omega_{\beta,{n}}$ is thereby directly connected to $\omega_{\infty,{n}}$ for each ${n}$.
Moreover, as we motivated at the end of Section \ref{sec:Setup} by the example of gBEC, the quasi-average method
is even indispensable for quantum continuous Bose-systems. An example of its use appears in the next
Section \ref{sec:QA-Boson}. See also the conclusion.

Note that for quantum continuous many-body systems or relativistic quantum field theory imposition of boundary
conditions is very questionable, or even not feasible.

The proof of (2) for the ferromagnet follows \cite{LSYng}, but using Bloch coherent states, instead of
Glauber coherent states, in the manner of Lieb's classic work on the classical limit of quantum spin systems
\cite{Lieb1}. It will not be spelled out here, because the next section will be devoted to a similar proof in
the case $G=U(1)$ and Boson systems, but we note the result:
\begin{proposition}\label{3.1}
The ergodic states $\omega_{\beta,{n}}$ in the decomposition (\ref{2.31}) may be obtained
by the Bogoliubov quasi-average method:
\begin{equation}\label{3.7}
\omega_{\beta, {n}} = \lim_{B \to +0} \lim _{V \to \infty} \omega_{\beta,\Lambda,{n}}
\end{equation}
where
\begin{equation}\label{3.8}
\omega_{\beta,\Lambda,{n}}(A) \equiv \frac{{\rm{Tr}}_{{\cal H}_{\Lambda}}
(\exp(-\beta H_{\Lambda,{B}})A)}{{\rm{Tr}}_{{\cal H}_{\Lambda}} \exp(-\beta H_{\Lambda,{B}})} \ ,
\end{equation}
with $A \in {\cal B}({\cal H}_{\Lambda})$, and $H_{\Lambda,{B}}$ is defined by (\ref{3.2}), (\ref{3.3})
for the ferromagnet (\ref{2.26}). The limit (\ref{3.7}) is taken along a (double) subsequence of the variables
$(B,V)$. For $A$ of the form (2.27), the actual double limit in (3.6) exists, and, if ODLRO holds in the
form (2.32), SSB in the form of definition 2.2 holds for the states (3.6).
\end{proposition}
{The identification of $\omega_{\beta,n}$ in (3.6) with those occurring in the decomposition (2.30) is possible by
the unicity of the ergodic decomposition, in view of the fact that the spin algebra is asymptotically abelian
for the space translations, see \cite{BR87}, pp 380,381. Since for KMS states the ergodic decomposition
coincides with the central decomposition, the extension of the states to elements of the center is also
unique.}
\section{{Continuous boson systems: quasi-averages, condensates, and pure states}}\label{sec:QA-Boson}

We now study the states of Boson systems, and, for that matter, assume, together with Verbeure (\cite{Ver}, Ch.4.3.2)
that they are analytic in the sense of \cite{BR97}, Ch.5.2.3. We start, with \cite{LSYng}, with the
Hamiltonian for Bosons in a cubic box $\Lambda$ of side $L$ and volume $V=L^{3}$,
\begin{equation}\label{4.1}
H_{\Lambda,\mu} = H_{0,\Lambda,\mu} + V_{\Lambda} \ ,
\end{equation}
where
\begin{equation}\label{4.2}
V_{\Lambda} = \frac{1}{V} \sum_{{k},{p},{q}}\nu({p})b_{{k}+{p}}^{*}b_{{q}-{p}}^{*}b_{{k}}b_{{q}} \ ,
\end{equation}
with periodic b.c., $\hbar=2m=1$, and $k,p,q \in \Lambda^*$. Here $\Lambda^*$ is dual (with respect to Fourier
transformation) set corresponding to $\Lambda$. Here $\nu$ is the Fourier transform of the two-body potential 
$v({x})$, with bound
\begin{equation}\label{4.3}
|\nu({k})| \le \phi < \infty \ ,
\end{equation}
and
\begin{equation}\label{4.4}
H_{0,\Lambda,\mu} = \sum_{{k}} {k}^{2} b_{{k}}^{*}b_{{k}} -\mu N_{\Lambda} \ ,
\end{equation}
\begin{equation}\label{4.5}
N_{\Lambda} = \sum_{{k}} b_{{k}}^{*}b_{{k}} \ ,
\end{equation}
with $[b_{{k}},b_{{l}}^{*}]=\delta_{{k},{l}}$ the second quantized annihilation and creation
operators. The quasi-Hamiltonian corresponding to (3.2) is taken to be
\begin{equation}\label{4.6}
H_{\Lambda,\mu,\lambda} = H_{\Lambda,\mu} + H_{\Lambda}^{\lambda} \ ,
\end{equation}
with the symmetry-breaking field analogous to (3.3) given by
\begin{equation}\label{4.7}
H_{\Lambda}^{\lambda} = \sqrt{V}(\bar{\lambda}_{\phi} b_{{0}}+\lambda_{\phi} b_{{0}}^{*}) \ .
\end{equation}
Above,
\begin{equation}\label{4.8}
\lambda_{\phi} = \lambda \exp(i\phi) \ \mbox{ with } \lambda \geq 0 \, , \,
\mbox{ where } \, {\rm{arg}}(\lambda) = \phi \in [0,2\pi) \ .
\end{equation}
We take initially $\lambda \geq 0$ and consider first the perfect Bose-gas to define
\begin{equation}\label{4.9.1}
H_{0,\Lambda, \mu, \lambda} = H_{0,\Lambda,\mu} + H_{\Lambda}^{\lambda} \ .
\end{equation}
We may write
\begin{equation*}
H_{0,\Lambda,\mu,\lambda}= H_{{0}}+H_{{k}\ne{0}} \ ,
\end{equation*}
where $H_{{0}} = -\mu \ b_{{0}}^{*}b_{{0}}+\sqrt{V}(\bar{\lambda}_{\phi} b_{{0}}+
\lambda_{\phi} b_{{0}}^{*})$. The grand partition function $\Xi_{\Lambda}$ splits into a product over the zero mode
and the remaining modes. We introduce the canonical shift transformation
\begin{equation}\label{4.9.2}
\widehat{b}_{{0}} := b_{{0}}+\frac{\lambda_{\phi} \sqrt{V}}{\mu} \ ,
\end{equation}
without altering the nonzero modes, and assume henceforth $\mu < 0$. We thus obtain for the grand partition
function $\Xi_{\Lambda}$,
\begin{equation}\label{4.9.3}
\Xi_{\Lambda}(\beta,\mu,\lambda) = (1-\exp(\beta \mu))^{-1}\exp(-\frac{\beta |\lambda|^{2}V}{\mu})\
\Xi^{\prime}_{\Lambda} \ ,
\end{equation}
where
\begin{equation}\label{4.9.4}
\Xi^{\prime}_{\Lambda} := \prod_{{k} \ne {0}}  (1-\exp(-\beta(\epsilon_{{k}}-\mu)))^{-1} \ ,
\end{equation}
with $\epsilon_{{k}}={k}^{2}$. Recall that the grand-canonical state for the
perfect Bose-gas is
\begin{equation}\label{4.9.5}
\omega^{0}_{\beta,\mu,\Lambda,\lambda}(\cdot):= {\frac{1}{\Xi_{\Lambda}}} \
{\rm{Tr}}[e^{-\beta H_{0,\Lambda,\mu,\lambda}} \ (\cdot )]  \ ,
\end{equation}
see Section \ref{sec:Setup}. Then it follows from (\ref{4.9.3})-(\ref{4.9.5}) that the mean density ${\rho}$
equals to
\begin{equation}
{\rho}=\omega_{\beta,\mu,\Lambda,\lambda}(\frac{N_{\Lambda}}{V})= \frac{1}{V(\exp(-\beta \mu)-1)}
+ \frac{|\lambda|^{2}}{\mu^{2}} + \\
\frac{1}{V} \sum_{{k} \ne {0}} \frac{1}{\exp(\beta(\epsilon_{{k}}-\mu))-1} \ .
\label{4.9.6}
\end{equation}

Equation (\ref{4.9.6}) is the starting point of our analysis. Let
\begin{equation}\label{4.10}
{\rho_{c}}(\beta) \equiv \int \frac{d{k}}{2\pi^{3}}(\exp(\beta \epsilon_{{k}})-1)^{-1} \ .
\end{equation}
\begin{lem}\label{4.1}
Let $0 < \beta <\infty$ be fixed. Then, for each
\begin{equation}\label{4.11}
{\rho_{c}} < {\rho} < \infty \ ,
\end{equation}
and for each $\lambda >0$, $V <\infty$, there exists a unique solution of (\ref{4.9.6}) of the form
\begin{equation}
\mu(V,|\lambda|,{\rho}) = -\frac{|\lambda|}{\sqrt{{\rho}-{\rho_{c}}(\beta)}}\\
 + \alpha(|\lambda|,V) \ ,
\label{4.12.1}
\end{equation}
with
\begin{equation}\label{4.12.2}
\alpha(|\lambda|,V) \ge 0 \ \ \forall \ |\lambda|, V \ ,
\end{equation}
and such that
\begin{equation}\label{4.13}
\lim_{|\lambda| \to 0} \lim_{V \to \infty} \frac{\alpha(|\lambda|,V)}{|\lambda|} = 0 \ .
\end{equation}
\end{lem}
\begin{rem}\label{4.2}
{We skip the proof of this lemma, but we note that besides the cube $\Lambda$, it is also true for the case of
three-dimensional anisotropic parallelepiped $\Lambda:= V^{\alpha_1}\times V^{\alpha_2}\times V^{\alpha_3}$,
with \textit{periodic boundary condition} (p.b.c.) and $\alpha_1 \geq \alpha_2 \geq \alpha_3$,
$\alpha_1 + \alpha_2 + \alpha_3 = 1$, i.e. the volume $|\Lambda| = V$.}
\end{rem}
We have now that
\begin{equation}
\lim_{\lambda \to +0} \lim_{V \to \infty}\omega^{0}_{\beta,\mu,\Lambda,\lambda}(\eta_{\Lambda}(b_{{0}}^{*}))=
\lim_{\lambda \to +0} \lim_{V \to \infty} \frac{\partial}{\partial \lambda_{\phi}} 
p_{\beta,\mu,\Lambda,\lambda_{\phi}}
\ ,
\label{4.14.1}
\end{equation}
where $\eta$ is defined as in (\ref{2.17}),(\ref{2.18}). Above we denote by
\begin{equation}\label{4.14.2}
 p_{\beta,\mu,\Lambda,\lambda}=\frac{1}{\beta V}\ln \Xi_{\Lambda}(\beta,\mu,\lambda) \ ,
\end{equation}
the pressure. By (\ref{4.9.6}),(\ref{4.14.2}) and the fact that the second term in (\ref{4.9.6}) equals
$(\lambda_{\phi}\bar{\lambda}_{\phi})/{\mu^{2}}$, we obtain
\begin{equation}\label{4.14.3}
\frac{\partial}{\partial \lambda_{\phi}} p_{\beta,\mu,\Lambda,\lambda_{\phi}}=
-\frac{\bar{\lambda}_{\phi}}{\mu} \ .
\end{equation}
By (\ref{4.14.3}), (\ref{4.12.1}) and (\ref{4.14.1}),
\begin{equation}
\lim_{\lambda \to +0} \lim_{V \to \infty}\omega^{0}_{\beta,\mu,\Lambda,\lambda}(\eta_{\Lambda}(b_{{0}}^{*}))
= \sqrt{\rho_{{0}}} \exp(i\phi) \ ,
\label{4.15.1}
\end{equation}
where, for the perfect Bose-gas,
\begin{equation*}
\rho_{{0}} = {\rho}-{\rho_{c}}(\beta) \ .
\end{equation*}
We see therefore that the phase in (\ref{4.14.1}) remains in (\ref{4.15.1}) even after the limit $\lambda \to +0$.
Define the states
\begin{equation}\label{PBG-LimSt-phi}
\omega^{0}_{\beta,\mu,\phi} := \lim_{\lambda \to +0}
\lim_{V \to \infty}\omega^{0}_{\beta,\mu,\Lambda,\lambda_{\phi}} \ ,
\end{equation}
where the double limit along a subnet exists by weak* compactness \cite{Hug}, \cite{BR87}.

For this and the forthcoming definitions, we are referring to the full interacting Bose gas (\ref{4.1})-(\ref{4.3}),
with $\omega$ replaced by $\omega^{0}$. The corresponding definitions
for the general case of the quantities $\omega_{\beta,\mu,\Lambda,\lambda}$,  $\omega_{\beta,\mu,\lambda,\phi}$ and
$\omega_{\beta, \mu, \phi}$ are the obvious analogues of (\ref{4.9.5}) and (\ref{PBG-LimSt-phi}),
with $H_{0,\Lambda,\mu,\lambda}$ replaced by $H_{\lambda,\mu,\lambda}$.

We say (cf Section \ref{sec:Setup}) that the interacting Bose-gas undergoes the \textit{zero-mode}
\textbf{Bose-Einstein condensation} (BEC) (and/or ODLRO) if
\begin{equation}\label{4.16}
\lim_{V \to \infty} \omega_{\beta,\mu,\Lambda}(\frac{b_{{0}}^{*}b_{{0}}}{V}) = \rho_{{0}}>0 \ .
\end{equation}

We define the group of \textit{gauge} transformations $\{\tau_{\lambda}|\lambda \in [0,2\pi)\}$ by the operations
\begin{equation}
\tau_{\lambda}(b^{*}(f)) = \exp(i\lambda)b^{*}(f)\\
\tau_{\lambda}(b(f)) = \exp(-i\lambda)b(f) \ ,
\label{4.17}
\end{equation}
where $b^{*}(f), b(f)$ are the creation and annihilation operators smeared over test-functions $f$ from the Schwartz
space. This group is isomorphic to the group $U(1)$.

Note that (\ref{4.15.1}), (\ref{PBG-LimSt-phi}) show that the states $\omega_{\beta,\mu,\phi}$ are not gauge
invariant. Assuming that they
are the ergodic states in the ergodic decomposition of $\omega_{\beta,\mu}$, which we shall prove next,
in greater generality, for the interacting system, it follows that BEC is equivalent to SSB for the free
Bose gas. It is illuminating to see, however, in the free case, a different explicit mechanism for the
appearance of the phase, which is connected with (\ref{4.12.1}) of Lemma \ref{4.1}, i.e., that the chemical
potential remains proportional to $|\lambda|$ even after the thermodynamic limit (together with (\ref{4.14.3})).
This property persists for the \textit{interacting} system, see below.
\begin{rem}\label{4.3}
{Note that these results are independent of the anisotropy, i.e. of whether the condensation for $\lambda =0$ is
in single mode ($k=0$) or it is extended as the gBEC-type III, Section \ref{sec:Setup}.
This means that the Bogoliubov quasi-average method solves the question about equivalence between
$\rm{(BEC)}_{qa}$, $\rm{(SSB)}_{qa}$ and $\rm{(ODLRO)}_{qa}$ if they are defined via \textit{one-mode} quasi-average.

To this aim we re-consider the prefect Bose-gas (\ref{G-C-PBG}) with symmetry breaking sources (\ref{4.7})
in a single mode $q \in \Lambda^{*}$:
\begin{eqnarray}\label{freeQE}
H^{0}_{\Lambda} (\mu; \eta) \, := \, H^{0}_{\Lambda}(\mu) \, + \, \sqrt{V} \ \big( \overline{\eta} \  b_{{q}} +
\eta \ b^{*}_{{q}} \big) \ , \ \mu < 0.
\end{eqnarray}
Then for a fixed density ${\rho}$, the the grand-canonical condensate equation (\ref{BEC-eq}) for (\ref{freeQE})
takes the following form:
\begin{eqnarray}\label{perfect-gas-with-source-density-equation-finite-volume}
&&{\rho} = \rho_{\Lambda}(\beta, \mu, \eta) \, := \, \frac{1}{V} \sum_{k \in \Lambda^{*}_{l}}
\omega_{\beta,\mu,\Lambda,\eta}^{0}(b^{*}_{k}b_{k}) = \\
&&\frac{1}{V} (e^{\beta(\varepsilon_{{q}} - \mu)}-1)^{-1} \, + \, \frac{1}{V}
\sum_{k\in \Lambda^{*}\setminus{q}} \frac{1}{e^{\beta(\varepsilon_{k} - \mu)}-1} \, + \, 
\frac{\vert \eta \vert\, ^{2}}
{(\varepsilon_{{q}} - \mu)\, ^{2}} \ . \nonumber
\end{eqnarray}

According the quasi-average method, to investigate a possible condensation, one must first take the thermodynamic 
limit in the right-hand side of (\ref{perfect-gas-with-source-density-equation-finite-volume}), and then switch 
off the symmetry breaking source: $\eta \rightarrow 0$. Recall that the critical density, which defines the 
threshold of boson saturation is equal to $\rho_c(\beta) = \mathcal{I}(\beta,\mu=0)$ (\ref{I}), where
$\mathcal{I}(\beta,\mu)=\lim_{\Lambda} \rho_{\Lambda}(\beta, \mu , \eta = 0)$.

Since $\mu < 0$, we have to distinguish two cases:\\
(i) Let ${q}\in \Lambda^{*}$ be such that $\lim_{\Lambda} \varepsilon_{{q}} > 0$, we obtain from
(\ref{perfect-gas-with-source-density-equation-finite-volume}) the condensate equation
\begin{eqnarray*}
{\rho} \, = \, \lim_{\eta \rightarrow 0}
\lim_{\Lambda} \rho_{\Lambda}(\beta, \mu, \eta) \, = \, \mathcal{I}(\beta, \mu) \ ,
\end{eqnarray*}
i.e. the quasi-average coincides with the average. Hence, we return to the analysis of the condensate equation
(\ref{perfect-gas-with-source-density-equation-finite-volume}) for $\eta =0$. This leads to finite-volume solutions
$\mu_{\Lambda}(\beta,\rho)$ and consequently to all possible types of condensation as a function of anisotropy
$\alpha_1$, see Section \ref{sec:Setup} for details.\\
(ii) On the other hand, if ${q}\in \Lambda^{*}$ is such that $\lim_{\Lambda} \varepsilon_{{q}} = 0$, then 
thermodynamic limit in the right-hand side of the condensate equation 
(\ref{perfect-gas-with-source-density-equation-finite-volume}) yields:
\begin{eqnarray}\label{perfect-gas-with-source-density-equation-infinite-volume}
{\rho} =  \lim_{\Lambda} \rho_{\Lambda}(\beta, \mu, \eta)
\, = \, \mathcal{I}(\beta, \mu) + \frac{\vert \eta \vert\, ^{2}}{\mu\, ^{2}} \ .
\end{eqnarray}

Now, if ${\rho} \leq \rho_{c}(\beta)$, then the limit of solution of
(\ref{perfect-gas-with-source-density-equation-infinite-volume}):
$\lim_{\eta \rightarrow 0}{\mu}(\beta, {\rho}, \eta) = {\mu}_{0} (\beta, {\rho}) <0$,
where ${\mu}(\beta,{\rho}, \eta)= \lim_{\Lambda}{\mu}_{\Lambda} (\beta,{\rho}, \eta)<0 $ is thermodynamic limit of
the finite-volume solution of condensate equation (\ref{perfect-gas-with-source-density-equation-finite-volume}).
Therefore, there is no condensation in any mode.

But if ${\rho} > \rho_{c}(\beta)$, then $\lim_{\eta \rightarrow 0}{\mu}(\beta, {\rho},\eta) =0$ and the
density of condensate is
\begin{equation}\label{BEC-qa}
\rho_{0}(\beta) = {\rho} - \rho_{c}(\beta) =
\lim_{\eta \rightarrow 0}\frac{\vert \eta \vert\, ^{2}}{\mu(\beta, {\rho},\eta)\, ^{2}} \ .
\end{equation}
Note that expectation of the particle density in the $q$-mode
(see (\ref{perfect-gas-with-source-density-equation-finite-volume})) is
\begin{equation*}
\omega_{\beta,\mu,\Lambda,\eta}^{0}({b^{*}_{q}b_{q}}/{V}) = \frac{1}{V} (e^{\beta(\varepsilon_{{q}} - \mu)}-1)^{-1}
+ \frac{\vert \eta \vert\, ^{2}} {(\varepsilon_{{q}} - \mu)\, ^{2}} \ .
\end{equation*}
Then by (\ref{BEC-qa}) the corresponding Bogoliubov quasi-average for ${b^{*}_{q}b_{q}}/{V}$ is equal to
\begin{eqnarray}\label{Bog-qa}
&&{\rho} - \rho_{c}(\beta)=\lim_{\eta \rightarrow 0}\lim_{\Lambda}\omega_{\beta,{\mu}_{\Lambda}
(\beta,{\rho}, \eta),\Lambda,\eta}^{0}({b^{*}_{q}b_{q}}/{V}) =  \\
&&\lim_{\eta \rightarrow 0}\lim_{\Lambda}\frac{1}{V} (e^{\beta(\varepsilon_{{q}} - {\mu}_{\Lambda}
(\beta,{\rho}, \eta))}-1)^{-1} +
\frac{\vert \eta \vert\, ^{2}} {(\varepsilon_{{q}} - {\mu}_{\Lambda} (\beta,{\rho}, \eta))\, ^{2}} \ , \nonumber
\end{eqnarray}
where ${\mu}_{\Lambda} (\beta,{\rho}, \eta)<0$ is a unique solution of the
condensate equation (\ref{perfect-gas-with-source-density-equation-finite-volume}) for ${\rho} > \rho_{c}(\beta) $.

Note that by virtue of (\ref{BEC-qa}) one has ${\mu}(\beta,{\rho}, \eta \neq 0)<0$. Hence, for any $k \neq q$
such that $\lim_{\Lambda} \varepsilon_{{k}} = 0$  we get
\begin{equation}\label{zero-non-zero-modes}
\lim_{\eta \rightarrow 0}\lim_{\Lambda}\omega_{\beta,{\mu}_{\Lambda}
(\beta,{\rho}, \eta),\Lambda,\eta}^{0}({b^{*}_{k}b_{k}}/{V}) =
\lim_{\eta \rightarrow 0}\lim_{\Lambda} \frac{1}{V}
\frac{1}{e^{\beta(\varepsilon_{{k}}- {\mu}_{\Lambda} (\beta,{\rho}, \eta)))}-1} = 0 \ ,
\end{equation}
i.e., for any $\alpha_1$ the quasi-average condensation $\rm{(BEC)}_{qa}$ occurs only in one mode (type I),
whereas for $\alpha_1 >1/2$ the BEC is of the type III, see Section \ref{sec:Setup}.

Similarly, diagonalisation  (\ref{4.9.2}) and (\ref{BEC-qa}) allow to apply the quasi-average method to calculate
a nonvanishing for ${\rho} > \rho_{c}(\beta)$ gauge-symmetry breaking $\rm{(SSB)}_{qa}$:
\begin{equation}\label{GSB-qa}
\lim_{\eta \rightarrow 0}\lim_{\Lambda}\omega_{\beta,{\mu}_{\Lambda}
(\beta,{\rho}, \eta),\Lambda,\eta}^{0}({b_{q}}/\sqrt{V}) =
\lim_{\eta \rightarrow 0}\frac{\eta}{\mu(\beta,{\rho}, \eta)} =
 e^{i \, {\rm{arg}}(\eta)} \, \sqrt{{\rho} - \rho_{c}(\beta)} \ ,
\end{equation}
along $\{\eta = |\eta| e^{i \, {\rm{arg}}(\eta)} \wedge |\eta|\rightarrow 0\}$.
Then by inspection of (\ref{Bog-qa}) and (\ref{GSB-qa}) we find that $\rm{(SSB)}_{qa}$ and $\rm{(BEC)}_{qa}$
are equivalent:
\begin{eqnarray}\label{Bog=GSB-qa}
&&\lim_{\eta \rightarrow 0}\lim_{\Lambda} \ \omega_{\beta,{\mu}_{\Lambda}
(\beta,{\rho}, \eta),\Lambda,\eta}^{0}({b^{*}_{q}}/\sqrt{V}) \ \omega_{\beta,{\mu}_{\Lambda}
(\beta,{\rho}, \eta),\Lambda,\eta}^{0}({b_{q}}/\sqrt{V}) = \\
&& = \lim_{\eta \rightarrow 0}\lim_{\Lambda} \ \omega_{\beta,{\mu}_{\Lambda}
(\beta,{\rho}, \eta),\Lambda,\eta}^{0}({b^{*}_{q}b_{q}}/{V}) =  {\rho} - \rho_{c}(\beta) \ . \nonumber
\end{eqnarray}

Note that by (\ref{PBG-ODLRO}) the $\rm{(SSB)}_{qa}$ and $\rm{(BEC)}_{qa}$  are in turn equivalent to
$\rm{(ODLRO)}_{qa}$, whereas for the conventional BEC on gets
\begin{equation*}
\lim_{\Lambda} \ \omega_{\beta,{\mu}_{\Lambda}
(\beta,{\rho}, \eta =0),\Lambda,\eta= 0}^{0}({b^{*}_{q}b_{q}}/{V}) =
\lim_{\Lambda} \ \omega_{\beta,{\mu}_{\Lambda}
(\beta,{\rho}, 0),\Lambda, 0}^{0}({b^{*}_{q}}/\sqrt{V}) \ \omega_{\beta,{\mu}_{\Lambda}
(\beta,{\rho}, 0),\Lambda, 0}^{0}({b_{q}}/\sqrt{V})= 0 \ ,
\end{equation*}
for any $\rho$ and $q\in \Lambda^{*}$ as soon as $\alpha_1 > 1/2$.
}
\end{rem}
We now consider the {{interacting case (4.1)-(4.5)}}. The famous \textbf{Bogoliubov approximation} of
replacing $\eta_{\Lambda}(b),\eta_{\Lambda}(b^{*})$ by $c$-numbers \cite{ZBru}, \cite{Za14} will be instrumental. 
It was proved by Ginibre \cite{Gin}, Lieb, Seiringer and Yngvason (\cite{LSYng1}, \cite{LSYng}) and
S\"{u}t\"{o} \cite{Suto1}, but we shall rely on the method of \cite{LSYng}, which uses the Berezin-Lieb
inequality \cite{Lieb1}.

Let $z$ be a complex number , $|z\rangle = \exp(-|z|^{2}/2 +z b_{{0}}^{*})|0\rangle$ the Glauber coherent vector
in ${\cal F}_{{0}}$ and, as in \cite{LSYng}, let
$(H_{\Lambda,\mu,\lambda})^{'}(z)$ be the \textit{lower symbol} of $H_{\Lambda,\mu,\lambda}$. Then
\begin{equation}
\exp(\beta V p_{\beta,\Lambda,\mu,\lambda}^{'})=\Xi_{\Lambda}(\beta,\mu,\lambda)^{'} =
\int d^{2}z {\rm{Tr}}_{{\cal H}^{'}} \exp(-\beta (H_{\Lambda,\mu,\lambda})^{'}(z)) \ ,
\label{eqn2.4.18}
\end{equation}
where ${\cal H}^{'}= {\cal F}_{{k} \ne {0}}$, with obvious notations for the Fock spaces associated to
the zero mode and the remaining modes. Consider the weight
\begin{equation}
{\cal W}_{\mu,\Lambda, \lambda}(z) := \Xi_{\Lambda}(\beta,\mu,\lambda)^{-1}\\
{\rm{Tr}}_{{\cal H}^{'}}\langle z| \exp(-\beta H_{\Lambda,\mu,\lambda})|z \rangle \ .
\label{eqn2.4.19}
\end{equation}
For almost all $\lambda >0$ it was proved in \cite{LSYng} that the density of distribution
${\cal W}_{\mu,\Lambda, \lambda} (\zeta \sqrt{V})$ converges, as $V \to \infty$, to a $\delta$ function at
the point $\zeta_{max}(\lambda)=\lim_{V \to \infty} {z_{max}(\lambda)}/{\sqrt{V}}$, where $ z_{max}(\lambda)$
maximizes the partition function ${\rm{Tr}}_{{\cal H}^{'}} \exp(-\beta (H_{\Lambda,\mu,\lambda})^{'}(z))$.
Although \cite{LSYng} took $\phi=0$ in (\ref{4.8}), their results in the general case (\ref{4.8}) may be obtained
by the trivial substitution $b_{{0}}\to b_{{0}}\exp(-i\phi)$, $b_{{0}}^{*} \to b_{{0}}^{*} \exp(i\phi)$ coming from
(\ref{4.6}). Note that their expression (34) in \cite{LSYng} may be thus re-written as
\begin{eqnarray}
&& \lim_{V \to \infty} \omega_{\beta,\mu,\Lambda,\lambda}(\eta_{\Lambda}(b_{{0}}^{*}\exp(i\phi))=
\lim_{V \to \infty} \omega_{\beta,\mu,\Lambda,\lambda}(\eta_{\Lambda}(b_{{0}}\exp(-i\phi))  \nonumber \\
&& = \zeta_{max}(\lambda)=\frac{\partial p(\mu,\lambda)}{\partial \lambda} \ , \label{eqn2.4.20}
\end{eqnarray}
and consequently
\begin{equation}
\lim_{V \to \infty} \omega_{\beta,\mu,\Lambda,\lambda}(\eta_{\Lambda}(b_{{0}}^{*})\eta_{\Lambda}(b_{{0}}))
= |\zeta_{max}(\lambda)|^{2} \ .
\label{eqn2.4.21}
\end{equation}
Here above,
\begin{equation}\label{4.22}
p(\beta,\mu,\lambda) = \lim_{V \to \infty} p_{\beta,\mu,\Lambda,\lambda} \ ,
\end{equation}
is the pressure in the thermodynamic limit. Equality (\ref{eqn2.4.20}) follows from the convexity of
$p_{\beta,\mu,\Lambda,\lambda}$ in $\lambda$ by the  Griffiths lemma \cite{Gri66}.
As it is shown in \cite{LSYng} the pressure $p(\beta,\mu,\lambda)$ is equal to
\begin{equation}\label{4.23}
p(\beta,\mu,\lambda)^{'} = \lim_{V \to \infty} p_{\beta,\mu,\Lambda,\lambda}^{'} \ .
\end{equation}
As well as it is also equal to the pressure $p(\beta,\mu,\lambda)^{''}$, which is the thermodynamic limit
of the pressure associated to the \textit{upper symbol} of $H_{\Lambda,\mu,\lambda}$.

It is crucial in the proof of \cite{LSYng} that all of these three pressures coincide with 
$p_{max}(\beta,\mu,\lambda)$, which is the pressure associated to
${\rm{max}}_{z} {\rm{Tr}}_{{\cal H}^{'}} \exp(-\beta (H_{\Lambda,\mu,\lambda})^{'}(z))$.

\begin{theo}\label{theo:4.1}
Consider the system of interacting Bosons (\ref{4.1})-(\ref{4.8}). If the system displays ODLRO in
the sense of (\ref{4.16}), the limit $\omega_{\beta,\mu,\phi}:=
\lim_{\lambda \to +0} \lim_{V \to \infty}\omega_{\beta,\mu,\Lambda,\lambda_{\phi}} $, on the
set $\{\eta(b_{{0}}^{*})^{m}\eta(b_{{0}})^{n}\}_{m,n=0,1}$ exists and satisfies
\begin{equation}\label{4.24.1}
\omega_{\beta,\mu,\phi} (\eta(b_{{0}}^{*})) = \sqrt{\rho_{0}} \exp(i\phi) \ ,
\end{equation}
\begin{equation}\label{4.24.2}
\omega_{\beta,\mu,\phi} (\eta(b_{{0}})) = \sqrt{\rho_{0}} \exp(-i\phi) \ ,
\end{equation}
together with
\begin{equation}\label{eqn3.4.24.3}
\omega_{\beta,\mu,\phi} (\eta(b_{{0}}^{*})\eta(b_{{0}}) = \omega_{\beta,\mu}((\eta(b_{{0}}^{*})\eta(b_{{0}}))
= \rho_{{0}} \ \ \forall \phi \in [0,2\pi) \ ,
\end{equation}
and
\begin{equation}\label{4.24.4}
\omega_{\beta,\mu} = \frac{1}{2\pi} \int_{0}^{2\pi} d\phi \ \omega_{\beta,\mu,\phi} \ .
\end{equation}
On the Weyl algebra the limit defining $\omega_{\beta,\mu,\phi}, \ \phi \in [0,2\pi)$
exists along a net in the $(\lambda,V)$ variables, and defines ergodic states coinciding with those
states that explicitly constructed in Theorem \ref{theo:A.1}. Conversely, if SSB occurs in the special sense
that (4.41) and (4.42) hold, with $\rho_{0} \ne 0$, then ODLRO in the sense of (4.25) takes place.
\end{theo}
\noindent
\begin{proof} We need only prove the direct statement, because the converse follows by applying the Schwarz
inequality to the states $\omega_{\beta,\mu,\phi}$, together with the forthcoming (\ref{eqn3.4.27}).

We thus prove ODLRO $\Rightarrow$ SSB. We first assume that some state $ \omega_{\beta,\mu,\phi_{0}},\phi_{0}
\in [0,2\pi)$ satisfies ODLRO. Then by (\ref{eqn2.4.21}),
\begin{equation}\label{eqn2.4.25.1}
\lim_{\lambda \to +0} \lim_{V \to \infty} \omega_{\beta,\mu,\Lambda,\lambda}(\eta(b_{{0}}^{*})\eta(b_{{0}}))
= \lim_{\lambda \to +0} |\zeta_{max}(\lambda)|^{2} =: \rho_{{0}} > 0 \ .
\end{equation}
The above limit exists by the convexity of $p(\mu,\lambda)$ in $\lambda$ and (\ref{4.14.1}) by virtue of
(\ref{eqn2.4.25.1}),
\begin{equation}\label{4.25.2}
\lim_{\lambda \to +0} \frac{\partial p(\mu,\lambda)}{\partial \lambda} \ne 0 \ .
\end{equation}
At the same time, (\ref{eqn2.4.20}) shows that all states $\omega_{\beta,\mu,\phi}$ satisfy (\ref{eqn2.4.25.1}).
Thus, SSB is broken in the states $\omega_{\beta,\mu,\phi},\phi \in [0,2\pi)$.\,
We now prove that the original assumption (\ref{4.16}) implies that all states
$\omega_{\beta,\mu,\phi},\phi \in [0,2\pi)$ exhibit ODLRO.

Gauge invariance of $\omega_{\beta,\mu,\Lambda}$ (or equivalently $H_{\Lambda,\mu}$) yields, by (\ref{4.7}),
(\ref{4.17}),
\begin{equation}
\omega_{\beta,\mu,\Lambda,\lambda}(\eta(b_{{0}}^{*})\eta(b_{{0}}))
=\omega_{\beta,\mu,\Lambda,-\lambda}(\eta(b_{{0}}^{*})\eta(b_{{0}})) \ .
\label{eqn2.4.26.1}
\end{equation}
Again by (\ref{4.7}), (\ref{4.12.1}) and gauge invariance of $H_{\Lambda,\mu}$,
\begin{equation*}
\lim_{\lambda \to -0} \frac{\partial p(\mu,\lambda)}{\partial \lambda}=
-\lim_{\lambda \to +0} \frac{\partial p(\mu,\lambda)}{\partial \lambda} \ ,
\end{equation*}
and, since by convexity the derivative ${\partial p(\mu,\lambda)}/{\partial \lambda}$ is monotone increasing, we find
\begin{equation}
\lim_{\lambda \to +0} \frac{\partial p(\mu,\lambda)}{\partial \lambda}
= \lim_{\lambda \to +0} \zeta_{max}(\lambda) = \sqrt{\rho_{0}} \ ,
\label{eqn2.4.26.2}
\end{equation}
\begin{equation}
\lim_{\lambda \to -0} \frac{\partial p(\mu,\lambda)}{\partial \lambda}
= -\lim_{\lambda \to +0} \zeta_{max}(\lambda)= -\sqrt{\rho_{0}} \ .
\label{eqn2.4.26.3}
\end{equation}
Again by (\ref{eqn2.4.26.1}),
\begin{equation}
\lim_{\lambda \to -0}\lim_{V \to \infty} \omega_{\beta,\mu,\Lambda,\lambda}(\eta(b_{{0}}^{*})\eta(b_{{0}}))
= \lim_{\lambda \to +0}\lim_{V \to \infty} \omega_{\beta,\mu,\Lambda,\lambda}(\eta(b_{{0}}^{*})\eta(b_{{0}})) \ .
\label{eqn2.4.26.4}
\end{equation}
By \cite{LSYng}, the weight ${\cal W}_{\mu,\lambda}$ is, for $\lambda=0$, supported on a disc with radius
equal to the right-derivative (\ref{4.25.2}). Convexity of the pressure as a function of $\lambda$ implies
\begin{eqnarray*}
\frac{\partial p(\mu,\lambda_{0}^{-})}{\partial \lambda_{0}^{-}} \le \lim_{\lambda \to -0}\frac{\partial
p(\mu,\lambda)}{\partial \lambda}
\le \lim_{\lambda \to +0}\frac{\partial p(\mu,\lambda)}{\partial \lambda} \le \frac{\partial p(\mu,
\lambda_{0}^{+})}{\partial \lambda_{0}^{+}} \ ,
\end{eqnarray*}
for any $\lambda_{0}^{-}<0<\lambda_{0}^{+}$. Therefore, by the Griffiths lemma (see e.g. \cite{Gri66},
\cite{LSYng}) one gets
\begin{equation}
\lim_{\lambda \to -0}\lim_{V \to \infty} \omega_{\beta,\mu,\Lambda,\lambda}(\eta(b_{{0}}^{*})\eta(b_{{0}}))
\le \lim_{V \to \infty} \omega_{\beta,\mu,\Lambda}(\frac{b_{{0}}^{*}b_{{0}}}{V})
\le \lim_{\lambda \to +0}\lim_{V \to \infty} \omega_{\beta,\mu,\Lambda,\lambda}(\eta(b_{{0}}^{*})\eta(b_{{0}})) \ .
\label{eqn3.4.27}
\end{equation}
Then (\ref{eqn2.4.26.4}) and (\ref{eqn3.4.27}) yield
\begin{equation}
\lim_{V \to \infty} \omega_{\beta,\mu,\Lambda}(\frac{b_{{0}}^{*}b_{{0}}}{V})=\\
\lim_{\lambda \to +0}\lim_{V \to \infty} \omega_{\beta,\mu,\Lambda,\lambda}(\eta(b_{{0}}^{*})\eta(b_{{0}})) \ \
\forall \phi \in [0,2\pi) \ .
\label{eqn3.4.28}
\end{equation}
This proves that all $\omega_{\beta,\mu,\phi},\phi \in [0,2\pi)$ satisfy ODLRO, as asserted.

By (\ref{eqn2.4.20}) and (\ref{eqn2.4.26.2}) one gets (\ref{4.24.1}) and (\ref{4.24.2}). Then (\ref{4.24.4}) is a
consequence of the gauge-invariance of $\omega_{\beta,\mu}$. Ergodicity of the states
$\omega_{\beta,\mu,\phi},\phi \in [0,2\pi$ follows from (\ref{eqn3.4.28}) and (\ref{4.24.1}), (\ref{4.24.2}).

An equivalent construction is possible using the Weyl algebra instead of the polynomial algebra, see
\cite{Ver}, pg. 56 and references given there for theorem 6.1 and similarly we could have proceeded so here.
The limit along a subnet in the $(\lambda,V)$ variables exists by weak* compactness, and, by asymptotically
abelianness of the Weyl algebra  for space translations
(see, e.g., \cite{BR97}, Example 5.2.19), the ergodic decomposition
(\ref{4.24.4}), which is also a central decomposition, is unique. Thus, the
$\omega_{\beta,\mu,\phi},\phi \in [0,2\pi)$ coincide with the states constructed in Theorem 6.1.
\end{proof}
\begin{rem}\label{4.1}
{Our Remark \ref{4.3} and Theorem \ref{theo:4.1} elucidate a problem discussed in \cite{LSYng}.
In this paper the authors defined a generalised  Gauge Symmetry Breaking via quasi-average $\rm{(GSB)}_{qa} \ $,
i.e. by $\lim_{\lambda \to +0} \lim_{V \to \infty} \omega_{\beta,\mu,\Lambda,\lambda}(\eta_{\Lambda}(b_{{0}})) \ne 0$.
(If it involves other than gauge group, we denote this by $\rm{(SSB)}_{qa}$.)
Similarly they modified definition of the one-mode condensation denoted by $\rm{(BEC)}_{qa}$ (\ref{eqn2.4.25.1}),
and established the equivalence: $\rm{(GSB)}_{qa} \Leftrightarrow \rm{(BEC)}_{qa}$.
They asked whether $\rm{(BEC)}_{qa} \Leftrightarrow \rm{BEC}$ ?
We show that $\rm{(GSB)}_{qa}$ coincides with GSB  (Definition 2.2), and that BEC is indeed equivalent
to $\rm{(BEC)}_{qa}$.
}
\end{rem}
\begin{rem}\label{4.2}
The states $\omega_{\beta,\mu,\phi}$ in Theorem \ref{theo:4.1} have the property ii) of Theorem \ref{theo:A.1},
i.e., if $\phi_{1} \ne \phi_{2}$, then
$\omega_{\beta,\mu,\phi_{1}} \ne \omega_{\beta,\mu,\phi_{2}}$. By a theorem of Kadison \cite{Kadison}, two factor
states are either disjoint or quasi-equivalent (see Remark 3.1 and references given there), and thus the states
$\omega_{\beta,\mu,\phi}$ for different $\phi$ are mutually disjoint. This fact has a simple explanation: only
for a finite system is the Bogoliubov transformation (\ref{4.9.2}) (which also applies to the interacting system),
which connects different $\phi$, unitary: for infinite systems one has to make an infinite change of an
extensive observable $b_{{0}} \sqrt{V}$, and mutually disjoint sectors result.
This phenomenon also occurs with regard to the magnetization in quantum spin systems, in correspondence to
(\ref{2.31}) and it is in this sense that the word ''degeneracy'' must be understood (compare with the
discussion in \cite{Bog70}).
\end{rem}

\section{Concluding remarks}

In this paper, we reexamined the issue of ODLRO versus SSB by the method of Bogoliubov quasi-averages, commonly
regarded as a \textit{symmetry-breaking trick}. We showed that it represents a general method of construction of
extremal, pure or ergodic states, both for quantum spin systems (Proposition \ref{3.1}) and many-body Boson systems
(Theorem \ref{theo:4.1}).
 The breaking of gauge symmetry in the latter has some analogy with the breaking of gauge and $\gamma_{5}$
 invariance in the Schwinger model (quantum electrodynamics of massless electrons in two dimensions) (\cite{LSwi}),
 in which the vacuum state decomposes in a manner similar to (\ref{4.24.4}). We believe, and argued so in
 Section \ref{sec:QA-spin}, that the quasi-average method is the only universally applicable method, in particular
 to relativistic quantum field theory, to which the imposition of classical boundary conditions is bound to be
 inconsistent with the general principles of local quantum theory, as in the case of the Casimir effect \cite{KNW}.

A general necessary feature for the applicability of the Bogoliubov method is the existence of an order parameter.
In the two examples treated, the Heisenberg ferromagnet (Section \ref{sec:Setup}, see also remark 2.5 concerning
order parameters for quantum spin systems in the general case) and many-body Boson systems
(Section \ref{sec:QA-Boson}),
the respective symmetry-breaking fields (\ref{3.3}) and (\ref{4.7}) are qualitatively different.
Note that (\ref{3.3}) commutes with $H_{\Lambda}$ and the corresponding order parameter, the magnetization,
is physically measurable.
Whereas (\ref{4.7}) does not commute with $H_{\Lambda}- \mu N_{\Lambda}$ (even in the free-gas case!), and the order
parameter involves a \textbf{phase} by (\ref{4.24.1}),(\ref{4.24.2}), which, at first glance, is not physically
measurable.
It has been observed, however, in the interference of two condensates of different phases \cite{ATMDKK} , \cite{BZ},
in the case of trapped gases.
In the latter case, Condensation takes place at ${k} \ne {0}$, and the version of Theorem \ref{theo:4.1} due to
Pul\`{e} et al \cite{PVZ} is the relevant one. Finally, in the quantum spin case there is a residual
symmetry (Remark \ref{rem:2.1}, but none, of course, in the Boson case. These remarks exemplify the rather wide
diversity of types of the \textit{Bogoliubov quasi-avearge}, which make its conjectured universal applicability
further plausible, see e.g. random boson systems \cite{JaZ10}.

As remarked by Swieca \cite{SwiecaJ}, it is the fluctuations occurring all over space which do not allow to take
the "charge" (e.g. (\ref{2.27})) in the limit $V \to \infty$ as a well-defined operator (this would, in particular,
contradict (\ref{2.24})), even if a meaning has been given to the density - as in (\ref{2.25}) - see also
Remark \ref{rem:2.3}. The additional input we offer is that the fluctuation of the charge density (or of a related
operator) is precisely a very nontrivial condition of ODLRO ((\ref{2.33}) or (\ref{4.16}) respectively).

As a final question, the treatment of the free Bose gas suggests that the chemical potential $\mu(\lambda) < 0$
for $\lambda \ne 0$ even after the thermodynamic limit also for interacting systems. It should be interesting to look
at Bose gases with repulsive interactions \cite{BraRo} from the point of view of quasi-average: $\rm{(SSB)}_{qa}$,
using the symmetry breaking term (\ref{4.7}).

\section{Appendix A}

In this Appendix we reproduce, for the reader's convenience, the statement of the basic theorem of Fannes,
Pul\`{e} and Verbeure \cite{FPV1}, see also \cite{PVZ} for the extension to nonzero momentum, and Verbeure's
book \cite{Ver}. Unfortunately, neither \cite{FPV1} nor \cite{PVZ} show that the states
$\omega_{\beta,\mu,\phi},\phi \in [0,2\pi)$ in the theorem below are ergodic. The simple, but instructive
proof of this fact was given by Verbeure in his book \cite{Ver}.
\begin{theo}\label{theo:A.1}
Let $\omega_{\beta,\mu}$ be an analytic, gauge-invariant equilibrium state. If
$\omega_{\beta,\mu}$ exhibits ODLRO (\ref{4.16}), then there exist ergodic states
$ \omega_{\beta,\mu,\phi},\phi \in [0,2\pi)$, not gauge invariant, satisfying
(i) $\forall \theta,\phi \in [0,2\pi)$ such that $\theta \ne \phi$, $\omega_{\beta,\mu,\phi} \ne
\omega_{\beta,\mu,\theta}$;
(ii) the state $\omega_{\beta,\mu}$ has the decomposition
\begin{equation*}
\omega_{\beta,\mu} = \frac{1}{2\pi} \int_{0}^{2\pi} d\phi\omega_{\beta,\mu,\phi} \ .
\end{equation*}
(iii) For each polynomial $Q$ in the operators $\eta(b_{{0}})$,$\eta(b_{{0}}^{*})$, and for each
$\phi \in [0,2\pi)$,
\begin{eqnarray*}
 \omega_{\beta,\mu,\phi}(Q(\eta(b_{{0}}^{*}),\eta(b_{{0}})X)
=  \omega_{\beta,\mu,\phi}(Q(\sqrt{\rho_{0}} \exp(-i\phi),\sqrt{\rho_{0}} \exp(i \phi)X)\ \
\forall X \in {\cal A} \ .
\end{eqnarray*}
\end{theo}

We remark, with Verbeure \cite{Ver}, that the proof of Theorem A.1 is \textbf{constructive}. One essential
ingredient is the separating character (or faithfulness) of the state $\omega_{\beta,\mu}$, i.e.,
$\omega_{\beta,\mu}(A) = 0$ implies $A=0$. This property, which depends on the extension of $\omega_{\beta,\mu}$
to the von-Neumann algebra $\pi_{\omega}({\cal A})^{''}$ (see \cite{BR97}, \cite{Hug}) is true for thermal
states, but is not true for ground states, even without this extension: in fact, a ground state (or vacuum)
is non-faithful on ${\cal A}$ (see proposition 3 of \cite{Wrep}). We see, therefore, that thermal states and
ground states might differ with regard to the ergodic decomposition (ii). Compare also with our discussion
in the Concluding remarks.

\bigskip

\noindent
\textbf{Acknowledgements}

\noindent
Some of the issues dealt with in this paper originate in the open problem posed
in Sec.3 of \cite{SeW} and at the of \cite{JaZ10}. One of us (W.F.W.) would like to thank G. L. Sewell for 
sharing with him his views on ODLRO along several years. He would also like to thank the organisers of the 
Satellite conference "Operator Algebras and Quantum Physics" of the XVIII conference of the IAMP (Santiago de 
Chile) in S\~{a}oPaulo, July 17th-23rd 2015, for the opportunity to present a talk in which some of the ideas 
of the present paper were discussed. We are  thankful to Bruno Nachtergaele for very  useful remarks, suggestions, 
and corrections, which greatly improved and clarified the paper.


\begin{thebibliography}{vdBJTLP86}

\bibitem[KK97]{ATMDKK}
M.~R. Andrews {,} C. G. Townsend {,} H. J. Miesner {,} D. S. Durfee {,} D.~M.
  Kurn and W.~Ketterle.
\newblock {\em Science}, \textbf{275}:637, 1997.

\bibitem[PiMe]{PiMe}
In S.~Attal, A.~Joye, and C.~A. Pillet, editors, {\em Open quantum systems
  {I}}. Lecture Notes in Mathematics 1880, Springer Verlag, 2006.

\bibitem[BZ]{BZ}
M. Beau and V.A. Zagrebnov.
\newblock The second critical density and anisotropic generalised condensation.
\newblock {\em Condensed Matter Physics}, \textbf{13}: 23003 (2010).

\bibitem[Bog07]{Bog07}
N.~N. Bogoliubov.
\newblock {\em Collection of scientific works in twelve volumes. Statistical Mechanics, volume 8 }:
Theory of Nonideal Bose Gas, Superfluidity and Superconductivity.
\newblock Moscow - Nauka , 2007.

\bibitem[Bog70]{Bog70}
N.~N. Bogoliubov.
\newblock {\em Lectures on quantum statistics, volume 2: Quasi-Averages}.
\newblock Gordon and Breach Sci. Publ., 1970.

\bibitem[Bog10]{Bog10}
N.~N. Bogoliubov and N.~N.~Bogoliubov Jr.
\newblock {\em Introduction to quantum statistical mechanics - 2nd. ed.}
\newblock World Scientific Publ. Co., 2010.

\bibitem[BR80]{BraRo}
O.~Bratelli and D.~W. Robinson.
\newblock Equilibrium states of a {Bose} gas with repulsive interactions.
\newblock {\em Australian J. Math. B}, \textbf{22}:129, 1980.

\bibitem[BR87]{BR87}
O.~Bratteli and D.~W. Robinson.
\newblock {\em Operator algebras and quantum statistical mechanics I}.
\newblock Springer, 1987.

\bibitem[BR97]{BR97}
O.~Bratteli and D.~W. Robinson.
\newblock {\em Operator algebras and quantum statistical mechanics II}.
\newblock Springer, 2nd edition, 1997.

\bibitem[MW13]{MWB}
Domingos H.~U. Marchetti and Walter~F. Wreszinski.
\newblock {\em Asymptotic Time Decay in Quantum Physics}.
\newblock World Scientific, 2013.

\bibitem[PV82]{FPV1}
M.~Fannes {,} J.~V. Pul\`{e} and A.~Verbeure.
\newblock On {B}ose condensation.
\newblock {\em Helv. Phys. Acta}, \textbf{55}:391--399, 1982.

\bibitem[Gin68]{Gin}
J.~Ginibre.
\newblock On the asymptotic exactness of the {B}ogoliubov approximation for
  many boson systems.
\newblock {\em Comm. Math. Phys.},\textbf{ 8}:26--51, 1968.

\bibitem[Gri66]{Gri66}
R.~B. Griffiths.
\newblock Spontaneous magnetization in idealized ferromagnets.
\newblock {\em Phys. Rev.}, \textbf{152}:240, 1966.

\bibitem[Haa96]{Haag}
R.~Haag.
\newblock {\em Local quantum physics - Fields, particles, algebras}.
\newblock Springer Verlag, 1996.



\bibitem[Hug72]{Hug}
N.~M. Hugenholtz.
\newblock States and representations in statistical mechanics.
\newblock In R.~F. Streater, editor, {\em Mathematics of Contemporary Physics}.
  Academic Press, 1972.
  
\bibitem[JaZ10]{JaZ10}
Th.~Jaeck and V.~A. Zagrebnov.
\newblock Exactness of the Bogoliubov approximation in random external potentials.
\newblock \textit{J.Math.Phys.} \textbf{51}, 123306:1-16 (2010).  

\bibitem[Kad62]{Kadison}
R.~V. Kadison.
\newblock States and representations. 
\newblock \textit{Trans. Amer. Math. Soc.,} \textbf{103}, 304–-319, 1962.

\bibitem[Swi70]{SwiecaJ}
J. A. Swieca.
\newblock Goldstone theorem and related topics.
\newblock In D.~Kastler, editor, {\em Cargese lectures in physics volume 4}. Gordon and
  Breach, 1970.

\bibitem[KNW07]{KNW}
N.~Kawakami, M.~C. Nemes, and W.~F. Wreszinski.
\newblock The {C}asimir effect for parallel plates revisited.
\newblock {\em J. Math. Phys.}, \textbf{48}:102302, 2007.

\bibitem[Lie73]{Lieb1}
E.~H. Lieb.
\newblock The classical limit of quantum spin systems.
\newblock {\em Comm. Math. Phys.}, \textbf{31}:327--340, 1973.

\bibitem[DLS]{DLS}
E.~H.~Lieb F.~J.~Dyson and B.~Simon.
\newblock Phase transitions in quantum spin systems with isotropic and
  non-isotropic interactions.
\newblock {\em Jour. Stat. Phys.}, \textbf{18}:335--383, 1978.

\bibitem[SY05]{LSYng1}
E.~H. Lieb {,}~R. Seiringer and J.~Yngvason.
\newblock Justification of $c$-number substitutions in {B}osonic {H}amiltonians.
\newblock {\em Phys. Rev. Lett.}, \textbf{94}:080401, 2005.

\bibitem[SY07]{LSYng}
E.~H. Lieb {,}~R. Seiringer and J.~Yngvason.
\newblock Bose {E}instein condensation and spontaneous symmetry breaking.
\newblock {\em Rep. Math. Phys.}, \textbf{59}:389, 2007.


\bibitem[LS71]{LSwi}
J.~H. Lowenstein and J.~A. Swieca.
\newblock Quantum electrodynamics in two dimensions.
\newblock {\em Ann. Phys.}, \textbf{68}:172--195, 1971.

\bibitem[Ntg]{Ntg}
B.~Nachtergaele.
\newblock Quantum spin systems after {DLS} 1978.
\newblock {\em Proc. Symp. Pure Math.}, \textbf{76}:47--68, 2007.

\bibitem[VZ05]{PVZ}
J.~V. Pul\`{e} {,}~A. Verbeure and V.~Zagrebnov.
\newblock On nonhomogeneous {B}ose condensation.
\newblock {\em Jour. Math. Phys.}, \textbf{46}:083301, 2005.

\bibitem[Rue69]{Ru}
D.~Ruelle.
\newblock {\em Statistical Mechanics - Rigorous Results}.
\newblock W. A. Benjamin Inc., 1969.

\bibitem[Sew86]{Sewell1}
G.~L. Sewell.
\newblock {\em Quantum theory of collective phenomena}.
\newblock Oxford University press, 1986.

\bibitem[Sew02]{Sewell}
G.~L. Sewell.
\newblock {\em Quantum mechanics and its emergent macrophysics}.
\newblock Princeton University Press, 2002.

\bibitem[SW09]{SeW}
G.~L. Sewell and W.~F. Wreszinski.
\newblock On the mathematical theory of superfluidity.
\newblock {\em J. Phys. A Math. Theor.}, \textbf{42}:015207, 2009.

\bibitem[Sim93]{SimonSM}
B.~Simon.
\newblock {\em The statistical mechanics of lattice gases vol.1}.
\newblock Princeton University Press, 1993.

\bibitem[S\"05]{Suto1}
A.~S\"{u}t\"{o}.
\newblock Equivalence of {B}ose {E}instein condensation and symmetry breaking.
\newblock {\em Phys. Rev. Lett.}, \textbf{94}:080402, 2005.

\bibitem[vdBLP]{vdBLP}
M.~van~den Berg {,} J. T.~Lewis and J.~V. Pul\`{e}.
\newblock A general theory of {B}ose {E}instein condensation.
\newblock {\em Helv. Phys. Acta}, \textbf{59}:1271, 1986.

\bibitem[Ver11]{Ver}
A.~Verbeure.
\newblock {\em Many body {B}oson systems - half a century later}.
\newblock Springer, 2011.


\bibitem[WA09]{WreA}
W.~F. Wreszinski and E.~Abdalla.
\newblock A precise formulation of the third law of thermodynamics.
\newblock {\em J. Stat. Phys.}, \textbf{134}:781--792, 2009.


\bibitem[Wre87]{Wrepp}
W.~F. Wreszinski.
\newblock Charges and symmetries in quantum theories without locality.
\newblock {\em Fortschr. der Physik}, \textbf{35}:379--413, 1987.

\bibitem[Wre05]{Wrep}
W.~F. Wreszinski.
\newblock Passivity of ground states of quantum systems.
\newblock {\em Rev. Math. Phys.}, \textbf{17}:1--14, 2005.

\bibitem[Za14]{Za14}
V.~A. Zagrebnov.
\newblock The Bogoliubov $c$-Number Approximation for Random Boson Systems.
\newblock \textit{Proceedings of the Kiev Institute of Mathematics,} Vol. 11(1), 123--140 , 2014.

\bibitem[ZB01]{ZBru}
V.~A. Zagrebnov and J.~B. Bru.
\newblock The {B}ogoliubov model of weakly imperfect {B}ose gas.
\newblock {\em Phys. Rep.}, \textbf{350}: 291--434, 2001.


\end{thebibliography}

\end{document}